\documentclass[11pt,a4paper]{article}
\pdfoutput=1

\usepackage[a4paper, margin=1in]{geometry}

\usepackage[utf8]{inputenc}
\usepackage{mathptmx}

\usepackage{amsmath}
\usepackage{amsfonts}
\usepackage{amssymb}
\usepackage{graphicx}
\usepackage{float}
\usepackage{amsthm}
\usepackage{enumitem}
\usepackage{xcolor}

\usepackage{microtype}

\usepackage{graphicx}

\usepackage{alltt}
\usepackage{xspace}

\usepackage{caption}
\usepackage[subrefformat=parens]{subcaption}

\usepackage{tikz}
\usetikzlibrary{fit}
\usetikzlibrary{automata}
\usetikzlibrary{calc}
\usetikzlibrary{positioning}
\usetikzlibrary{matrix}
\usepackage{hyperref}

\pgfdeclarelayer{background}
\pgfsetlayers{background,main}

\usepackage{listings}

\tikzstyle{lcm}=[rectangle,draw=black!50,fill=white,text width=3.5mm,align=center]
\tikzstyle{schedot}=[coordinate]
\tikzstyle{sched}=[coordinate,rectangle,text width=3.5mm]
\tikzstyle{cycle}=[rectangle,rounded corners=5pt,fill=black!10]
\tikzset{
		>=latex,
		node distance=0.2\columnwidth
}

\newcommand\etal{\mbox{\xspace\emph{et al.}}\xspace}
\newcommand\CENTRALIZED{\mbox{Centralized}\xspace}
\newcommand\FSYNC{\mbox{FSYNC}\xspace}
\newcommand\SSYNC{\mbox{SSYNC}\xspace}
\newcommand\ASYNC{\mbox{ASYNC}\xspace}
\newcommand\ASYNCLC{\mbox{LC-atomic \ASYNC}\xspace}
\newcommand\ASYNCMOVE{\mbox{Move-atomic \ASYNC}\xspace}

\newcommand\WAIT{\mbox{\textsc{Wait}}\xspace}
\newcommand\LOOK{\mbox{\textsc{Look}}\xspace}
\newcommand\COMPUTE{\mbox{\textsc{Compute}}\xspace}
\newcommand\BEGMOVE{\ensuremath{\mbox{\textsc{Move}}_{\mbox{\textsc{b}}}}\xspace}
\newcommand\ENDMOVE{\ensuremath{\mbox{\textsc{Move}}_{\mbox{\textsc{e}}}}\xspace}
\newcommand\MOVE{\mbox{\textsc{Move}}\xspace}

\newcommand\Look{\mbox{\textsc{L}}\xspace}
\newcommand\Compute{\mbox{\textsc{C}}\xspace}
\newcommand\Begmove{\ensuremath{\mbox{\textsc{B}}}\xspace}
\newcommand\Endmove{\ensuremath{\mbox{\textsc{E}}}\xspace}

\newcommand\STAY{\mbox{\textsc{Stay}}\xspace}
\newcommand\OTHER{\mbox{\textsc{M2O}}\xspace}
\newcommand\HALF{\mbox{\textsc{M2H}}\xspace}
\newcommand\MISS{\mbox{\textsc{Miss}}\xspace}
\newcommand\BOT{\ensuremath{\bot}\xspace}

\newcommand\FAR{\mbox{F{\scriptsize{AR}}}\xspace}
\newcommand\NEAR{\mbox{N{\scriptsize{EAR}}}\xspace}
\newcommand\SAME{\mbox{S{\scriptsize{AME}}}\xspace}

\newcommand\BLACK{\mbox{\textsc{Black}}\xspace}
\newcommand\WHITE{\mbox{\textsc{White}}\xspace}
\newcommand\RED{\mbox{\textsc{Red}}\xspace}
\newcommand\YELLOW{\mbox{\textsc{Yellow}}\xspace}
\newcommand\GREEN{\mbox{\textsc{Green}}\xspace}

\newtheorem{Lemma}{Lemma}

\newtheorem{Theorem}{Theorem}
\newtheorem{Definition}{Definition}

\title{Using Model Checking to Formally Verify Rendezvous Algorithms for Robots with Lights in Euclidean Space%
    \thanks{%
            Research partly supported by JST SICORP and JSPS KAKENHI Grant No.\,17K00019.
    }
}

\author{
    Xavier Défago$^a$, Adam Heriban$^b$, Sébastien Tixeuil$^b$, and Koichi Wada$^c$
    \medskip
    \\ $^a$School of Computing, Tokyo Institute of Technology, Japan
    \\ $^b$Sorbonne Université, CNRS, LIP6, Paris, France
    \\ $^c$Faculty of Science and Engineering, Hosei University, Tokyo, Japan
}

\date{}

\graphicspath{{Figs/}}

\begin{document}

\maketitle

\begin{abstract}
The paper details the first successful attempt at using model-checking techniques to verify the correctness of distributed algorithms for robots evolving in a \emph{continuous} environment. The study focuses on the problem of rendezvous of two robots with lights.

There exist many different rendezvous algorithms that aim at finding the minimal number of colors needed to solve rendezvous in various synchrony models (\emph{e.g.}, \FSYNC, \SSYNC, \ASYNC). While these rendezvous algorithms are typically very simple, their analysis and proof of correctness tend to be extremely complex, tedious, and error-prone as impossibility results are based on subtle interactions between robots activation schedules.

The paper presents a generic verification model written for the SPIN model-checker. In particular, we explain the subtle design decisions that allow to keep the search space finite and tractable, as well as prove several important theorems that support them. As a sanity check, we use the model to verify several known rendezvous algorithms in six different models of synchrony. In each case, we find that the results obtained from the model-checker are consistent with the results known in the literature. The model-checker outputs a counter-example execution in every case that is known to fail.

In the course of developing and proving the validity of the model, we identified several fundamental theorems, including the ability for a well chosen algorithm and \ASYNC scheduler to produce an emerging property of memory in a system of oblivious mobile robots, and why it is not a problem for luminous rendezvous algorithms.
\end{abstract}

\section{Introduction}
Since the seminal work of Suzuki and Yamashita \cite{SuzukiY99}, much research on cooperative mobile robots has aimed at identifying the minimal assumptions (in terms of synchrony, sensing capabilities, environment, etc.) under which basic problems such as gathering or rendezvous can be solved.

Robots are modelled as mathematical points in the 2D Euclidean plane that independently execute their own instance of the same algorithm.
In the model we consider, robots are anonymous (\emph{i.e.}, indistinguishable from each-other), oblivious (\emph{i.e.}, no persistent memory of the past is available), and disoriented (\emph{i.e.}, they do not agree on a common coordinate system). The robots operate in Look-Compute-Move cycles. In each cycle a robot ``Looks'' at its surroundings and obtains (in its own coordinate system) a snapshot containing the locations of all robots. Based on this visual information, the robot ``Computes'' a destination location (still in its own coordinate system) and then ``Moves'' towards the computed location. Since the robots are identical, they all follow the same deterministic algorithm. The algorithm is oblivious if the computed destination in each cycle depends only on the snapshot obtained in the current cycle (and not on the past history of execution). The snapshots obtained by the robots are not consistently oriented in any manner (that is, the robots' local coordinate systems do not share a common direction nor a common chirality%
\footnote{Chirality denotes the ability to distinguish left from right.}%
).

The execution model significantly impacts the ability to solve collaborative tasks. Three different levels of synchronization have been commonly considered. The strongest model is the fully-synchronous (\FSYNC) model~\cite{SuzukiY99} where each phase of each cycle is performed simultaneously by all robots. The semi-synchronous (\SSYNC) model~\cite{SuzukiY99} considers that time is discretized into rounds, and that in each round an arbitrary yet non-empty subset of the robots are active. The robots that are active in a particular round perform exactly one atomic Look-Compute-Move cycle in that round. It is assumed that the scheduler (seen as an adversary) is fair in the sense that in each execution, every robot is activated infinitely often. The weakest model is the asynchronous  (\ASYNC) model~\cite{DBLP:journals/tcs/FlocchiniPSW05,2012Flocchini}, which allows arbitrary delays between the Look, Compute and Move phases, and the movement itself may take an arbitrary amount of time. 


\subsection{Gathering}
The gathering problem is one of the benchmarking tasks in mobile robot networks, and has received a considerable amount of attention (\emph{e.g.}, \cite{DBLP:journals/siamcomp/AgmonP06,DBLP:journals/trob/AndoOSY99,DBLP:conf/icdcs/Bouzid0T13,DBLP:conf/algosensors/CiceroneSN14,DBLP:journals/siamcomp/CohenP05,DBLP:conf/wdag/DefagoGMP06,DBLP:journals/tcs/FlocchiniPSW05,HeribanDT18,DBLP:journals/siamcomp/IzumiSKIDWY12,DBLP:conf/opodis/OkumuraWD18,DBLP:journals/taas/SouissiDY09,SuzukiY99,2019Flocchini}).
The gathering task consists in all robots reaching a single point, not known beforehand, in finite time. The particular case of gathering two robots is called \emph{rendezvous} in the literature.
A foundational result~\cite{SuzukiY99} shows that in the \SSYNC model, no deterministic algorithm can solve rendezvous without additional assumptions. This impossibility result naturally extends to the \ASYNC model~\cite{DBLP:journals/tcs/FlocchiniPSW05}. 

To circumvent the aforementioned impossibility result, it was proposed to endow each robot with a \emph{light}~\cite{0001FPSY16}, that is, it is capable of emitting one color among a fixed number of available colors, visible to all other robots. This additional capacity allows to solve rendezvous in the most general \ASYNC model, provided that lights of robots are capable to emit at least \emph{four} colors. In the more restricted \SSYNC model, Viglietta~\cite{Viglietta13} proved that being able to emit two colors is sufficient to solve the rendezvous problem. In the same paper~\cite{Viglietta13}, Viglietta proves~\cite{Viglietta13} that three colors are sufficient in \ASYNC. Both solutions in \ASYNC~\cite{0001FPSY16,Viglietta13} and \SSYNC~\cite{Viglietta13} output a correct behavior independently of the initial value of the lights' colors. Recently, Okumura\etal~\cite{OkumuraWK17} presented an rendezvous algorithm with two colors in \ASYNC assuming \emph{rigid} moves (that is, the move of every robot is never stopped by the scheduler before completion), or assuming non-rigid moves but robots being aware of $\delta$ (the minimum distance before which the scheduler cannot interrupt their move). Also, the solution of Okumura\etal~\cite{OkumuraWK17} requires lights to have a specific color in the initial configuration. Finally, Heriban\etal~\cite{HeribanDT18} show that two colors are necessary and sufficient in \ASYNC without extra assumptions. 

However, all aforementioned results were proved in an ad-hoc manner (\emph{i.e.}, with handwritten proofs). As they often use case-based reasoning, are lengthy and complex, they are tedious to check and thus to write, hence error-prone. 

\subsection{Related work}
Designing and proving mobile robot protocols is notoriously difficult. Formal methods encompass a long-lasting path of research that is meant to overcome errors of human origin. Unsurprisingly, this mechanized approach to protocol correctness was used in the context of mobile robots~\cite{BonnetDPPT14,DevismesLPRT12,BerardLMPTT16,AugerBCTU13,MilletPST14,CourtieuRTU15, berard15infsoc,RubinZMA15,DevismesLPRT12,SangnierSPT17,BalabonskiPRT18}.

When robots move freely in a continuous two-dimensional Euclidean space (as considered in this paper), to the best of our knowledge the only formal framework available is Pactole.\footnote{\url{http://pactole.lri.fr}}
It relies on higher-order logic to certify impossibility results~\cite{AugerBCTU13,CourtieuRTU15,BalabonskiPRT18}, as well as the correctness of algorithms~\cite{CourtieuRTU16,DevismesLPRT12} in the \FSYNC and \SSYNC models, possibly for an arbitrary number of robots (hence in a scalable manner). Pactole was recently extended by Balabonski\etal~\cite{BalabonskiCPRTU18} to handle the \ASYNC model, thanks to its modular design. However, in its current form, Pactole lacks automation; that is, in order to prove a result formally, one still has to write the proof (that is automatically verified), which requires expertise both in Coq (the language Pactole is based upon) and about the mathematical and logical arguments one should use to complete the proof. 

On the other side, model checking and its derivatives (automatic program synthesis, parameterized model checking) hint at more automation once a suitable model has been defined with the input language of the model checker. 
In particular, model-checking proved useful to find bugs (usually in the \ASYNC setting)~\cite{BerardLMPTT16,DoanBO16,Doan0017} and to formally check the correctness of published algorithms~\cite{DevismesLPRT12,BerardLMPTT16,RubinZMA15}. Automatic program synthesis~\cite{BonnetDPPT14,MilletPST14} was used to obtain automatically algorithms that are ``correct-by-design''. However, those approaches are limited to instances with few robots. Generalizing them to an arbitrary number of robots with similar models is doubtful as Sangnier\etal~\cite{SangnierSPT17} proved that safety and reachability problems are undecidable in the parameterized case. Another limitation of the above approaches is that they \emph{only} consider cases where mobile robots \emph{evolve in a \textbf{discrete} space} (\emph{i.e.}, graph). This limitation is due to the model used, that closely matches the original execution model by Suzuki and Yamashita~\cite{SuzukiY99}. As a computer can only model a finite set of locations, a continuous 2D Euclidean space cannot be expressed in this model.

Overall, the only way to obtain automated proofs of correctness in the continuous space context through model checking is to use a more abstract model.

\subsection{Our contribution}
We present a new formal model and related proof techniques enabling the use of model checking tools for mobile robots evolving in a 2D Euclidean space. Because of the nature of rendezvous, we abstract the Euclidean space into a single line, which is itself discretized into a finite number of well-defined distance states. The number of those states depends on the exact hypotheses of the verified algorithm. Our process relies on the fact that rigid-motion algorithms are, by definition, distance-independent, and so limiting the space model to two states: \{gathered, non-gathered\} is sufficient to properly model the distance. We also prove that, in the case of self-stabilizing algorithms, rigid and non-rigid motions are equivalent. Proving non-rigid, non-self-stabilizing algorithms is, however, more complex, and we propose a method to check them properly. We propose a way to prove algorithms such as the one presented in Okumura\etal~\cite{OkumuraWK17}.

We also show that proving any algorithm for the \ASYNC scheduler is, unsurprisingly, extremely difficult. We prove that for a well-chosen scheduler, algorithm, and initial configuration, it is possible for every configuration in the execution to depend on the initial configuration even when robots are oblivious. However, we point out a structural property satisfied by known rendezvous algorithms that enables reducing the verification space tremendously. 

\section{Definitions and Terminology}
\subsection{System model}

The system consists of a finite set of dimensionless robots evolving in a boundless 2D Euclidean space, devoid of any landmarks or obstacles.
The robots cannot explicitly communicate with each other and do not share any notion of a global coordinate system as defined by a common origin, unit distance, or directions and orientations of the axes. These coordinate systems are not assumed to be consistent during the execution.

Robots are \emph{anonymous} in that they are unaware of an identity and execute the exact same algorithm consisting of cycles of the basic operations: \LOOK, \COMPUTE, \MOVE. We also consider a \WAIT phase between the end of the \MOVE phase and the beginning of the \LOOK phase.
Robots are \emph{oblivious} in the sense that their new destination only depends on the current snapshot.

Robots' activations are independent and governed by an adversary scheduler acting according to a synchrony model which can vary from \CENTRALIZED (where activation cycles execute in mutual exclusion) or \FSYNC (fully synchronous; where activation cycles always occur in parallel) to \SSYNC (semi-synchronous, where the scheduler can choose between \FSYNC and \CENTRALIZED) and \ASYNC (asynchronous; where activation steps occur independently). We also consider the more recent \ASYNCLC and \ASYNCMOVE schedulers~\cite{OkumuraWK17}, which are variants of the \ASYNC for which the \LOOK and \COMPUTE, or the entire \MOVE phase, respectively, must happen atomically.
The scheduler may also have the ability to influence the movement of robots. In the case of \emph{rigid} motion, a robot always reaches its target upon completion of its \MOVE phase. In the case of \emph{non-rigid} motion, the scheduler may stop the robot before it reaches its target, but not before it has travelled a minimum distance of $\delta>0$, usually unknown to the robots.

We also consider a fair scheduler, which activates every robot infinitely often.

Additionally, in this paper we consider that robots are equipped with a light that can emit a color among a fixed number of different colors. A robot can observe the color of all lights during its \LOOK phase and change the color of its light at the end of its \COMPUTE phase. During the compute phase, the snapshot may lead to a deterministic change in color. Then, we say the new color is \emph{pending} until the \COMPUTE phase has ended. Similarly, we call \emph{target} (or pending move) the destination dictated by a robot's snapshot that it tries to reach in its next \MOVE phase. More specifically, this is called the "\emph{full light}" model. Other variants in the literature consider the cases where a robot can only observe the light of the other robots (\emph{external light} model), or can only observe its own light (\emph{internal light} model). This paper considers both the full and external light models.

A more detailed model is provided in the appendix.

\subsection{Configurations and executions}

The union of the local states (position, color, phase, pending move and color) 
of all robots defines a \emph{configuration}. An \emph{execution} is a sequence, possibly infinite, starting in an initial configuration and where each transition corresponds to the activation of a subset of the robots according to the constraints of the scheduler.

\subsection{Rendezvous problem}

The \emph{gathering} problem requires that all robots eventually reach the same location from any initial location and regardless of activations decided by the scheduler.

The \emph{rendezvous} problem is another name for the gathering of two robots. Intuitively, the problem may seem simpler to solve due to the smaller number of robots, but this is actually the opposite, due to symmetry. Indeed, with only two robots and no lights, all configurations are symmetrical and hence convey no information other than distance. Because of the lack of a common coordinate system, this information is limited to the binary \{gathered, non-gathered\}.

\subsection{Self-Stabilization}

A rendezvous algorithm is self-stabilizing if robots eventually reach and stay forever at the same location regardless of the \emph{initial configuration}. Algorithms that set constraints on the initial configuration (\emph{e.g.}, must start with a specific color) are not self-stabilizing.
We introduce a more refined definition of self-stabilization.

\begin{Definition}[\emph{Simple} Self-Stabilization]
An algorithm is \emph{simply self-stabilizing} for problem $\mathcal{P}$ if it solves $\mathcal{P}$ starting from any initial position and any initial color, with all robots in the \WAIT phase.
\end{Definition}

\begin{Definition}[\emph{Complete} Self-Stabilization]
An algorithm is \emph{completely self-stabilizing} for problem $\mathcal{P}$ if it solves $\mathcal{P}$ from any initial position, color, phase, target and pending color.
\end{Definition}

Following the same terminology, an initial configuration where both robots are in the \WAIT phase is called a \emph{simple} initial configuration. Otherwise it is a \emph{complete} initial configuration. We similarly define \emph{simple} executions and \emph{complete} executions.
If a \emph{complete} execution has a common suffix with a \emph{simple} execution, we say it is \emph{simple-reachable}.

\section{From the system model to the verification model}
\label{sec:theorems}

Implementing the system model we consider into a \emph{verification model} that can be checked by a model-checker is difficult, as some elements are continuous (position of both robots, pending moves of both robots). This section is dedicated to proving that those problematic elements can be discretized in a way that enables mechanized verification. 

\subsection{\emph{Simple} vs. \emph{Complete} self-stabilization}

This subsection is dedicated to proving that pending moves and pending colors can be removed from the verification model in the case of self-stabilizing rendezvous algorithms. This is true for all self-stabilizing algorithms under the \FSYNC, \CENTRALIZED, and \SSYNC schedulers (Lemma~\ref{Lem:SimToComSSYNC}), and true for specific self-stabilizing algorithms under the \ASYNC scheduler (Theorem~\ref{Thm:SoNice}). 

\begin{Lemma}
\label{Lem:ComToSim}
Any completely self-stabilizing algorithm is also simply self-stabilizing.
\end{Lemma}

To prove the reverse, we want to prove that every \emph{complete} execution is simple-reachable. If this is the case, then any \emph{complete} execution is eventually explored by starting from a \emph{simple} initial configuration.

\begin{Lemma}
\label{Lem:SimToComSSYNC}
Under the \FSYNC, \CENTRALIZED and \SSYNC schedulers, any simply self-stabilizing algorithm is also completely self-stabilizing.
\end{Lemma}

Intuitively, it seems that this also holds for the case of \ASYNC algorithms. Since both robots are oblivious, with the exception of color, it seems logical that the system eventually ``forgets'' its initial configuration, and becomes reachable from a \emph{simple} initial configuration.

Surprisingly, we show that it is possible, for a well-chosen algorithm, \ASYNC scheduling, and \emph{complete} initial configuration to create an infinite execution that never becomes \emph{simple-reachable}.

\begin{Theorem}
\label{Thm:ObliviousMemory}
A simply self-stabilizing algorithm is not necessarily completely self-stabilizing in the \ASYNC model.
\end{Theorem}

For a well-chosen algorithm, scheduler, and \emph{complete} initial configuration, it is possible for the system to have an emerging property of memory. This is because it is possible to have the current configuration depend on the initial configuration indefinitely (see appendix for details).

In practice, most rendezvous algorithms in the literature prevent this behavior by forcing robot~$A$ to wait for robot~$B$ without being able to change its color or move. As the scheduler relies on its ability to feed outdated information infinitely often to the robot performing the \LOOK to create non-oblivious executions, when synchronization is enforced by the algorithm, the scheduler looses this capability.

We first observe that in the case of oblivious deterministic rendezvous algorithms (that is, algorithms that make use of a single color), simple self-stabilization implies complete self-stabilization.

\begin{Theorem}
\label{Thm:targetSS}
If a simply self-stabilizing algorithm achieves rendezvous, it is also a simply self-stabilizing algorithm for rendezvous \emph{with arbitrary initial targets}. 
\end{Theorem}

Because no deterministic solution exists in this setting, we study the case of \ASYNC luminous rendezvous algorithms using at least two colors.

To tackle the case of color memory, we consider two luminous \ASYNC rendezvous algorithms:
Viglietta 3~\cite{Viglietta13} colors and Heriban 2 colors~\cite{HeribanDT18}. For both algorithms, we find a structural condition we prove is sufficient to block the scheduler from creating any color memory execution.

\begin{Theorem}[Identical Color Condition]
\label{Thm:SoNice}
We define the \emph{identical color condition} (ICC) as: "For any pair of robots $A$ and $B$ whose colors are $Color_A$ and $Color_B$, respectively, $A$ can decide on a new color (different from $Color_A$) if and only if its snapshot shows $Color_A$ and $Color_B$ are identical."
For any luminous rendezvous algorithm, if the algorithm satisfies ICC and is simply self-stabilizing, then is it also completely self-stabilizing. \end{Theorem}

In this paper, the model checker only uses \emph{simple} initial configurations. From Theorem~\ref{Thm:SoNice}, we can extend the positive results to \emph{complete} initial configurations when considering algorithms that satisfy ICC such as Viglietta 3-colors~\cite{Viglietta13} and Heriban 2 colors~\cite{HeribanDT18}. Of course, negative results are not impacted, as a counter-example in the simply self-stabilizing context is also a counter-example in the completely self-stabilizing context.

\subsection{Self-Stabilization and rigidity}

The non-rigid assumption is another source of a continuous variable in the model: when the robot target a point at some distance $d\geq\delta>0$, the scheduler may stop the robot anywhere between $\delta$ and $d$. In this section, we explore under which circumstances we can restrict the verification model to rigid moves only without loosing generality, and show that completely self-stabilizing rendezvous algorithms satisfy the condition (Theorem~\ref{Thm:rigid}).  

Using criteria such as (complete) self-stabilization and rigidity, we can define four different settings for rendezvous algorithms according to the combination of $\{\textit{rigid}, \textit{non-rigid\}}$  and $\{\textit{self-stabilizing}, \textit{non}$ $\textit{self-sta\-bi\-li\-zing\}}$.
Studying the literature on rendezvous algorithms, we were not able to find examples of self-stabilizing algorithms requiring rigid moves that failed with non-rigid moves. The following theorem shows that such algorithms cannot exist.

\begin{Theorem}
\label{Thm:rigid}
Any completely self-stabilizing algorithm that achieves rendezvous assuming rigid moves also achieves it assuming non-rigid moves.
\end{Theorem}

Theorem~\ref{Thm:rigid} implies that in order to prove complete self-stabilization, it is only necessary to prove the property assuming rigid moves. 

\subsection{Proving rendezvous algorithms}

The last mile to the verification model is to show that the remaining continuous variable of the current configuration (the distance between the two robots) can be abstracted into two states only (Theorem~\ref{Thm:BinSpace})

We first show that, when robots may only move in straight lines, rendezvous algorithms need only use three kinds of movements: stay put, move to the midpoint, or move to the other robot.

\begin{Lemma}
\label{Lem:Mov}
The three types of motion stay put (\STAY), move to the midpoint (\HALF), and move to the other robot (\OTHER) are both necessary and sufficient to achieve rendezvous in \SSYNC and \ASYNC.
\end{Lemma}

\begin{Lemma}
\label{Lem:Conv}
When two robots $A$ and $B$ start from the \WAIT state, in an infinite fair execution, either the distance between $A$ and $B$ eventually remains constant, or it converges toward $0$.
\end{Lemma}

We observe that, for any rendezvous algorithm execution that solely uses the three required movements, and where robots $A$ and $B$ start in the \WAIT phase, the entire execution happens on the line $(AB)$.

\begin{Theorem}[Rigid motion model]
\label{Thm:BinSpace}
When proving the correctness of a rigid-motion rendezvous algorithm that solely uses the three required movements, only using the two model states \emph{gathered} and \emph{not-gathered} is sufficient to properly represent the Euclidean plane.
\end{Theorem}

Thanks to Theorem~\ref{Thm:BinSpace}, we now have a finite number of states to model the entire Euclidean plane in the case of rigid rendezvous. Note that this holds for both self-stabilizing and non-self-stabilizing algorithms. In turn, this implies that we may use model checking to verify the validity of a rendezvous algorithm in the particular case of rigid motion. Furthermore, we can verify (simple) self-stabilization by checking all possible pairs of colors, and that this also verifies complete self-stabilization if the algorithm satisfies ICC (by Theorem~\ref{Thm:SoNice}) and non-rigid completely self-stabilizing algorithms (by Theorem~\ref{Thm:rigid}).

The only remaining family of algorithms is non-rigid, non-self-stabilizing algorithms.

\begin{Theorem}
\label{Thm:sht}
To prove non-rigid non-self-stabilizing algorithms to achieve rendezvous, verifying rigid behavior is necessary but not sufficient to prove the correctness of the algorithm.
\end{Theorem}

We present a possible approach for checking these algorithms in the appendix. To the best of our knowledge, the 4-color external light algorithm by Okumura\etal~\cite{DBLP:conf/opodis/OkumuraWD18} is the only known rendezvous algorithm that satisfies these two criteria.

Note that our reasoning is only true in the case where the behavior of the algorithm is the same for any distance between $A$ and $B$ that is greater than zero. Recently, Okumura\etal~\cite{DBLP:conf/birthday/OkumuraWK18} introduced an algorithm under the additional assumption that robots have the knowledge of $\delta$. Because of this, the behavior of the algorithm is different when the distance is less than $\delta$, and when it is between $\delta$ and $2\delta$, which means that the rigid and the non-rigid behavior of the algorithm are different.
To prove such algorithms, we also need to both prove the rigid and non-rigid behaviors.

\begin{Theorem}
\label{Thm:wada}
To prove completely self-stabilizing rendezvous algorithms whose behavior differs depending whether the distance between $A$ and $B$ is smaller than $\delta$, between $\delta$ and $2\delta$, or greater than $2\delta$, it is sufficient to consider the case where robots are initially $3\delta$ apart. 
\end{Theorem}

\section{Verification Model}

\subsection{Position}

Our verification model only needs to consider two different positions (called \NEAR, \SAME) depending on the distance between the two robots. This choice is justified by the definition of the model (two robots, no shared coordinate system, no landmarks, oblivious robots) and Theorem~\ref{Thm:BinSpace}

\subsection{Activation and synchrony}
In the verification model, the activation cycle of a robot $r$ is defined as a sequence of four consecutive atomic events: \LOOK, \COMPUTE, \BEGMOVE, and \ENDMOVE. Each of the four events is as follows:

\begin{description}
\item[\LOOK (\Look)]
	The robot obtains a snapshot observation of the environment which consists of the color of both robots and the  location of the other robot with respect to $r$'s local coordinate system where $r$ is always at the origin.
\item[\COMPUTE (\Compute)]
	The robot executes the algorithm which is defined as a function of the latest observation that returns a new color and a move order. In the verification model, we assume that the light of the robot changes as part of the compute event.
	
\item[\BEGMOVE (\Begmove)]
	The robot begins moving according to the move order. Although it can be observed while moving, the actual position of the robot is actually undefined until the movement is completed with the \ENDMOVE event.
	
	However, we assume that the robot moves according to a straight line toward its target (point computed by the algorithm), and that it only progresses forward until reaching the destination of the move (point within the reachable distance $\delta$).
	
\item[\ENDMOVE (\Endmove)]
	The robot ends its move and has moved a distance of at least $\delta$ towards the target. If the distance to the target was equal or less than $\delta$, the robot has reached its target.

\end{description}

We use this four-events verification model instead of the classical three phases model because of the flexibility it allows when defining variants of the \ASYNC model. We can set the \LOOK phase to be instantaneous by linking a \COMPUTE event to happen right after each \LOOK event, or set to \ASYNCLC by linking \LOOK, \COMPUTE and \BEGMOVE, or \ASYNCMOVE by linking \BEGMOVE and \ENDMOVE, and so on.

\begin{Theorem}
\label{Thm:Scheduling}
The \FSYNC scheduler can be properly simulated by activating sequentially the \LOOK phase of robot $A$, the \LOOK phase of robot $B$, the \COMPUTE \BEGMOVE, and \ENDMOVE of robot $A$ and finally the \COMPUTE \BEGMOVE, and \ENDMOVE of robot $B$, infinitely often.
\end{Theorem}

\begin{Theorem}
\label{Thm:Events}
Two simultaneous events $E1$ and $E2$ can be properly simulated by exploring both sequences $(E1,E2)$ and $(E2,E1)$.
\end{Theorem}

\begin{Theorem}
\label{Thm:Fair}
Except for non-rigid non-self-stabilizing rendezvous, a fair scheduler can be properly simulated by an 8$N$-bounded scheduler, where $N$ denotes the number of colors available to the algorithm.
\end{Theorem}

Because we check for a maximum of 5 colors, the model checker uses a 40-bounded scheduler.

\newcommand\activation[2]{%
	\begin{pgfonlayer}{background}
		\path (#1.west |- #1.north)+(-1mm,1mm) node (a1) {};
		\path (#2.east |- #2.south)+(1mm,-1mm) node (a2) {};
		\path[cycle] (a1) rectangle (a2);
	\end{pgfonlayer}
}
\newcommand\scheduleline[2]{%
	\path (#1)+(0,5mm) node (a1) {};
	\path (#2)+(0,-5mm) node (a2) {};	
	\path[draw,dotted,thick] (a1) -- (a2);
}
\newcommand\scheduleDot[1]{%
	\draw[fill=black] (#1) circle (2pt);
}
\newcommand\scheduleAt[1]{%
	\path (#1)+(0,10mm) node (top) {};
	\path (#1)+(0,-10mm) node (bot) {};
	\draw[dotted,thick] (top) -- (bot);
	\scheduleDot{#1}
}
\newcommand\activateA[2]{%
	\draw[->] (#1.west) -- (#2.south west);
}
\newcommand\activateB[2]{%
	\draw[->] (#1.west) -- (#2.north west);
}

\begin{figure}
    \centering
    \begin{tabular}{cc}
    \begin{minipage}{0.61\linewidth}
    	\raggedright
    	\begin{tikzpicture}
	    	\matrix [row sep=2mm,column sep=.5mm] {
    			\node{\scriptsize Robot $A$};
	    		&
	    		& \node[lcm](Al){\Look};
	    		&
	    		& \node[lcm](Ac){\Compute};
    			&
			    & \node[lcm](Ab){\Begmove};
			    &
			    & %
			    &
			    & \node[lcm](Ae){\Endmove};
			    &
			    & %
			    &
			    & %
			    &
			    & \node[lcm](Al2){\Look};
			    &
			    & \node[lcm](Ac2){\Compute};
			    &
			    & \phantom{\node[lcm](Ab2){\Begmove};}%
			    &
			    \\
			    \node{\scriptsize scheduler};
			    & \node[schedot](sd0){};
			    & \node[sched](sc1){};
			    & \node[schedot](sd1){};
			    & \node[sched](sc2){};
			    & \node[schedot](sd2){};
			    & \node[sched](sc3){};
			    & \node[schedot](sd3){};
			    & \node[sched](sc4){};
			    & \node[schedot](sd4){};
			    & \node[sched](sc5){};
			    & \node[schedot](sd5){};
			    & \node[sched](sc6){};
			    & \node[schedot](sd6){};
			    & \node[sched](sc7){};
			    & \node[schedot](sd7){};
			    & \node[sched](sc8){};
			    & \node[schedot](sd8){};
			    & \node[sched](sc9){};
			    & \node[schedot](sd9){};
			    & \node[sched](sc10){};
			    & \node[schedot](sd10){};
			    \\
			    \node{\scriptsize Robot $B$};
			    &
			    & %
			    &
			    & %
			    &
			    & %
			    &
			    & \node[lcm](Bl){\Look};
			    &
			    & %
			    &
			    & \node[lcm](Bc){\Compute};
			    &
			    & \node[lcm](Bb){\Begmove};
			    &
			    & %
			    &
			    & %
			    &
			    & \node[lcm](Be){\Endmove};
			    &
			    \\			
		    };
		    \activation{Al}{Ae}
		    \activation{Bl}{Be}
		    \activation{Al2}{Ab2}
		    \draw (sd0) -- (sd10);
		    \scheduleAt{sd0}
		    \scheduleAt{sd1}
		    \scheduleAt{sd2}
		    \scheduleAt{sd3}
		    \scheduleAt{sd4}
		    \scheduleAt{sd5}
		    \scheduleAt{sd6}
		    \scheduleAt{sd7}
		    \scheduleAt{sd8}
		    \scheduleAt{sd9}
		    \scheduleAt{sd10}
		    \activateA{sc1}{Al}
		    \activateA{sc2}{Ac}
		    \activateA{sc3}{Ab}
		    \activateB{sc4}{Bl}
		    \activateA{sc5}{Ae}
		    \activateB{sc6}{Bc}
		    \activateB{sc7}{Bb}
		    \activateA{sc8}{Al2}
		    \activateA{sc9}{Ac2}
		    \activateB{sc10}{Be}
	    \end{tikzpicture}
	    \subcaption{ASYNC}
	    \label{fig:sched:ASYNC}
    \end{minipage}
    &
    \begin{minipage}{0.33\linewidth}
	    \raggedleft
        \begin{tikzpicture}
    	    \matrix [row sep=2mm,column sep=3mm] {
    		    \node[lcm,minimum height=4mm] (actS) {~};
    		    & \node[anchor=west] {\footnotesize activation step};
            \\
    		    \node (actC) {~};
    		    & \node[anchor=west] {\footnotesize activation cycle};
            \\
                \node (sch){};
                & \node[anchor=west] {\footnotesize non-deterministic choice};
            \\
            };
            \activation{actC}{actC}
            \scheduleDot{sch}
        \end{tikzpicture}
    \end{minipage}
    \end{tabular}
    \\\bigskip
    \begin{tabular}{c@{}c}
        \begin{minipage}{0.49\linewidth}
	        \raggedright
	        \begin{tikzpicture}
		        \matrix [row sep=2mm,column sep=0mm] {
			    \node{\scriptsize Robot $A$};
			        & \node (Asched0){};
			        & \node[lcm](Al){\Look};
			        & \node[lcm](Ac){\Compute};
			        & \node[lcm](Ab){\Begmove};
			        & \node[lcm](Ae){\Endmove};
			        & \node (Asched1){};
			        &
			        &
			        &
			        &
			        & \node (Asched2){};
			        \\
			        \node{\scriptsize scheduler};
			        & \node[schedot](sd0){};
			        & \node[sched](sc1){};
			        & \node[sched](sc2){};
			        & \node[sched](sc3){};
			        & \node[sched](sc4){};
			        & \node[schedot](sd1){};
			        & \node[sched](sc5){};
			        & \node[sched](sc6){};
			        & \node[sched](sc7){};
			        & \node[sched](sc8){};
			        & \node[schedot](sd2){};
			        \\
			        \node{\scriptsize Robot $B$};
			        & \node (Bsched0){};
			        &
			        &
			        &
			        &
			        & \node (Bsched1){};
			        & \node[lcm](Bl){\Look};
			        & \node[lcm](Bc){\Compute};
			        & \node[lcm](Bb){\Begmove};
			        & \node[lcm](Be){\Endmove};
			        & \node (Bsched2){};
			        \\			
		        };
		        \activation{Al}{Ae}
		        \activation{Bl}{Be}
		        \draw (sd0) -- (sd2);
		        \scheduleAt{sd0}
		        \scheduleAt{sd1}
		        \scheduleAt{sd2}
		        \activateA{sc1}{Al}
		        \activateA{sc2}{Ac}
		        \activateA{sc3}{Ab}
		        \activateA{sc4}{Ae}
		        \activateB{sc5}{Bl}
		        \activateB{sc6}{Bc}
		        \activateB{sc7}{Bb}
		        \activateB{sc8}{Be}
	        \end{tikzpicture}
	        \subcaption{\CENTRALIZED}
        	\label{fig:sched:centralized}
	    \end{minipage}
	    &
	    \begin{minipage}{0.475\linewidth}
	        \raggedright
        	\begin{tikzpicture}
        		\matrix [row sep=2mm,column sep=0mm] {
        			\node{\scriptsize Robot $A$};
        			& \node (Asched0){};
        			& \node[lcm](Al){\Look};
        			&
        			& \node[lcm](Ac){\Compute};
        			& \node[lcm](Ab){\Begmove};
	        		& \node[lcm](Ae){\Endmove};
	        		&
	        		&
	        		& \phantom{\node[lcm](Aend){\Endmove};}
	        		& \node (Asched1){};
	        		\\
	        		\node{\scriptsize scheduler};
	        		& \node[schedot](sd0){};
	        		& \node[sched](sc1){};
	        		& \node[sched](sc2){};
	        		& \node[sched](sc3){};
	        		& \node[sched](sc4){};
	        		& \node[sched](sc5){};
	        		& \node[sched](sc6){};
	        		& \node[sched](sc7){};
	        		& \node[sched](sc8){};
	        		& \node[schedot](sd1){};
	        		\\
	        		\node{\scriptsize Robot $B$};
	        		& \node (Bsched0){};
	        		& \phantom{\node[lcm](Bstart){\Look};}
	        		& \node[lcm](Bl){\Look};
	        		&
	        		&
	        		&
	        		& \node[lcm](Bc){\Compute};
	        		& \node[lcm](Bb){\Begmove};
	        		& \node[lcm](Be){\Endmove};
	        		& \node (Bsched1){};
	        		\\			
	        	};
	        	\activation{Al}{Aend}
	        	\activation{Bstart}{Be}
	        	\draw (sd0) -- (sd1);
	        	\scheduleAt{sd0}
	        	\scheduleAt{sd1}
	        	\activateA{sc1}{Al}
	        	\activateB{sc2}{Bl}
	        	\activateA{sc3}{Ac}
	        	\activateA{sc4}{Ab}
	        	\activateA{sc5}{Ae}
	        	\activateB{sc6}{Bc}
	        	\activateB{sc7}{Bb}
	        	\activateB{sc8}{Be}
	        \end{tikzpicture}
        	\subcaption{\FSYNC}
        	\label{fig:sched:FSYNC}
        \end{minipage}
	\end{tabular}
	\caption{Simulation of main schedulers as a Promela process.}
	\label{fig:sched}
\end{figure}
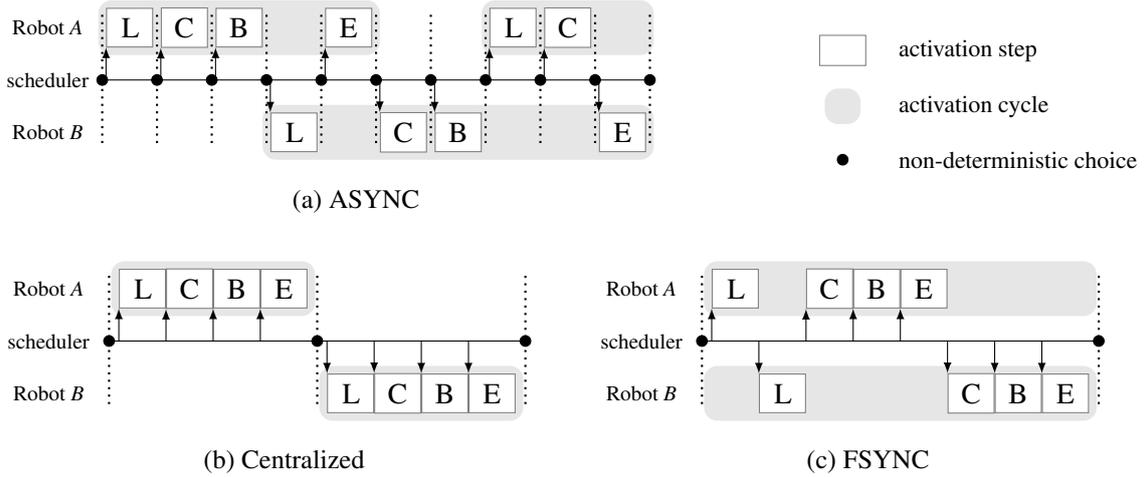

\subsection{Movement resolution}

The key to our verification model is the idea of movement resolution. When a robot completes its movement, this translates into a change of the verification model according to specific rules, which are described below.

\noindent\textbf{Stationary moves}
When the computed move is invariant or stationary (e.g., \OTHER when the observed position is \SAME), its pending move is systematically translated to an equivalent \STAY move. A robot that has a \STAY pending move is stationary, and hence is not observed as moving between \BEGMOVE and \ENDMOVE.

\noindent\textbf{From \NEAR or \SAME position}
The key aspect of the verification model is the case when the robots are in rigid motion (\NEAR and \SAME), and we only detail the resolution of moves for this combined case.

\begin{description}
\item[\STAY] No change.

\item[\MISS] The pending move is a sure miss. A miss happens for instance if $r$ observes the other robot while it moves. It also happens indirectly during movement resolution. The result of a \MISS move is always \NEAR (in particular, it can happen if the position was \SAME).

\item[\OTHER]
	If the position is \SAME, then the move is treated as a \STAY. If the other robot is either \STAY or \BOT, then the position is now \SAME. Else, if the other robot has a pending move, it is converted to a \MISS and the position is now \SAME.
	
\item[\HALF]
    If the pending move of the other robot is \STAY or \BOT, then neither the distance nor the pending move of the other robot changes. 
	If the pending move of the other robot is also \HALF, then the move is potentially successful and the pending move of the other robot is changed into an \OTHER move that targets the location just newly reached. Thus, provided that the first robot doesn't move in the meantime, the movement of the other robot later leads to \SAME.
	Else, the distance is still \NEAR and the pending move or the other robot is now \MISS.
\end{description}

\begin{table}
    \centering
    \caption{Movement resolution.
        Each tuple represents (position, me.pending, other.pending).
        Major changes appear in boldface.
        A bottom value (\BOT) represents the absence of pending moves.
        Wildcard (*) replaces any value and placeholder (--) retains the original value. Rule precedence is from top to bottom.
}
    \label{tab:movement}
    \small
    \renewcommand{\arraystretch}{0.9}
    \begin{tabular}{@{(}c@{, }c@{, }c@{) $\leadsto$ (}c@{, }c@{, }c@{)}}
        \hline
        \multicolumn{6}{l}{ from \NEAR or \SAME }\\
        *     & \STAY  & *              & --             & \BOT & -- \\
        *     & \MISS  & \STAY | \BOT   & \NEAR          & \BOT & -- \\
        *     & \MISS  & *              & \NEAR          & \BOT & \textbf{\MISS} \\
        *     & \OTHER & \STAY | \BOT   & \textbf{\SAME} & \BOT & -- \\
        \SAME & \OTHER & *              & \SAME          & \BOT & -- \\
        \NEAR & \OTHER & *              & \textbf{\SAME} & \BOT & \textbf{\MISS} \\
        *     & \HALF  & \HALF          & --             & \BOT & \textbf{\OTHER} \\
        *     & \HALF  & \STAY | \BOT   & --             & \BOT & -- \\
        *     & \HALF  & *              & --             & \BOT & \textbf{\MISS} \\
        \hline
    \end{tabular}
\end{table}

\section{Checking Rendezvous Algorithms}

\subsection{Rendezvous Algorithms}

To assess the verification model, we have checked three trivial baseline algorithms as well as seven known algorithms from the literature. For each of these algorithms, it is widely-known in which models they achieve rendezvous or fail. Unless explicitly stated otherwise, algorithms are non-rigid and self-stabilizing. The latter seven algorithms are detailed in Figure~\ref{fig:algos} and the code for some of them can be found in the appendix.

\begin{figure}
    \centering
    \begin{minipage}{0.45\linewidth}
      \centering
      \begin{subfigure}{\linewidth}
        \begin{tabular}{l@{\quad$\leadsto$\quad}l}
            (\BLACK, \BLACK) & \WHITE, \STAY \\
            (\BLACK, \WHITE) & skip \\
            (\WHITE, \BLACK) & --, \OTHER\\
            (\WHITE, \WHITE) & \BLACK, \HALF
        \end{tabular}
        \caption{Vig2Cols: 2 colors for \ASYNCLC \cite{Viglietta13}}
        \label{fig:algos:Vig2Cols}
      \end{subfigure}
      \\\bigskip
      \begin{subfigure}{\linewidth}
        \begin{tabular}{l@{\quad$\leadsto$\quad}l}
            (\BLACK, \BLACK) & \WHITE, \HALF \\
            (\BLACK, \WHITE) & --, \OTHER \\
            (\BLACK, \RED)   & skip \\
            (\WHITE, \BLACK) & skip \\
            (\WHITE, \WHITE) & \RED, \STAY \\
            (\WHITE, \RED)   & --, \OTHER \\
            (\RED, \BLACK)   & --, \OTHER \\
            (\RED, \WHITE)   & skip \\
            (\RED, \RED)     & \BLACK, \STAY
        \end{tabular}
        \caption{Vig3Cols: 3 colors for \ASYNC \cite{Viglietta13}}
        \label{fig:algos:Vig3Cols}
      \end{subfigure}
      \\\bigskip
      \begin{subfigure}{\linewidth}
        \begin{tabular}{l@{\quad$\leadsto$\quad}l}
            (\BLACK, \BLACK) & \WHITE, \STAY \\
            (\BLACK, \WHITE) & skip \\
            gathered         & skip \\
            (\WHITE, \BLACK) & --, \OTHER\\
            (\WHITE, \WHITE) & \BLACK, \HALF
        \end{tabular}
        \caption{Her2Cols: 2 colors for \ASYNC \cite{HeribanDT18}}
        \label{fig:algos:Opt2Cols}
      \end{subfigure}
    \end{minipage}
    \hfill
    \begin{minipage}{0.45\linewidth}
      \centering
      \begin{subfigure}{\linewidth}
        \begin{tabular}{l@{\quad$\leadsto$\quad}l}
            (*, \BLACK) & \WHITE, \HALF \\
            (*, \WHITE) & \RED, \STAY \\
            (*, \RED)   & \BLACK, \OTHER
        \end{tabular}
        \caption{Flo3ColsX: 3 colors external for \SSYNC\cite{DBLP:journals/tcs/FlocchiniSVY16}}
        \label{fig:algos:Ext3Cols}
      \end{subfigure}
      \\\bigskip
      \begin{subfigure}{\linewidth}
        \begin{tabular}{l@{\quad$\leadsto$\quad}l}
            (*, \BLACK) & \WHITE, \HALF \\
            (*, \WHITE) & \RED, \STAY \\
            (*, \RED)   & \YELLOW, \OTHER \\
            (*, \YELLOW)   & \GREEN, \STAY \\
            (*, \GREEN)   & \BLACK, \STAY
        \end{tabular}
        \caption{Oku5ColsX: 5 colors external for \ASYNCLC\cite{DBLP:conf/opodis/OkumuraWD18}}
        \label{fig:algos:Oku5}
      \end{subfigure}
      \\\bigskip
      \begin{subfigure}{\linewidth}
        \begin{tabular}{l@{\quad$\leadsto$\quad}l}
            (*, \BLACK) & \WHITE, \HALF \\
            (*, \WHITE) & \RED, \STAY \\
            (*, \RED)   & \YELLOW, \OTHER \\
            (*, \YELLOW)   & \BLACK, \STAY
        \end{tabular}
        \caption{Oku4ColsX: 4 colors external Quasi-Self-Stabilizing for \ASYNCLC\cite{DBLP:conf/opodis/OkumuraWD18}}
        \label{fig:algos:Oku4}
      \end{subfigure}
      \\\bigskip
      \begin{subfigure}{\linewidth}
        \begin{tabular}{l@{\quad$\leadsto$\quad}l}
            (*, \BLACK) & \WHITE, \HALF \\
            (*, \WHITE) & \RED, \STAY \\
            (*, \RED)   & \WHITE, \OTHER 
        \end{tabular}
        \caption{Oku3ColsX: 3 colors external rigid Non-Self-Stabilizing for \ASYNCLC\cite{DBLP:conf/opodis/OkumuraWD18}}
        \label{fig:algos:Oku3}
      \end{subfigure}
    \end{minipage}
    \\
    \bigskip
    \caption{Rendezvous algorithms from the literature. Guards match the (me.color, other.color). Wildcard (*) replaces any value, placeholder (--) retains the original value, "skip" means no change, "gathered" holds only when both robots have the same position. Rule precedence is from top to bottom.}
    \label{fig:algos}
\end{figure}

\begin{enumerate}
\item ``NoMove'': the robots never move and hence the algorithm always fails to achieve rendezvous regardless of the model;
\item ``ToHalf'': the robots always target the midpoint and the algorithm fails in all models but \FSYNC.
\item ``ToOther'': the robots always target the other robot's position and the algorithm fails in all models but \CENTRALIZED.
\item ``Vig2Cols'': the algorithm (Figure~\ref{fig:algos:Vig2Cols}) was originally proved correct in \SSYNC~\cite{Viglietta13} but was later proved to also achieve rendezvous in \ASYNCLC \cite{OkumuraWK17}. The algorithm is known to fail in \ASYNC \cite{Viglietta13}.
\item ``Vig3Cols'': the algorithm (Figure~\ref{fig:algos:Vig3Cols}) is known to succeed in \ASYNC and consequently in all other weaker models~\cite{Viglietta13}.
\item ``Her2Cols'': the algorithm (Figure~\ref{fig:algos:Opt2Cols}) is an extension of "Vig2Cols" that uses only two colors but succeeds in \ASYNC \cite{HeribanDT18}. The algorithm is optimal in the sense that rendezvous in \ASYNC cannot possibly be achieved with fewer colors.
\item ``Flo3ColsX'': the algorithm (Figure~\ref{fig:algos:Ext3Cols}) achieves rendezvous in \SSYNC in a model with external colors. The algorithm is known to succeed in \SSYNC and to fail in \ASYNC \cite{DBLP:journals/tcs/FlocchiniSVY16}.
\item ``Oku5colsX'': the algorithm achieves rendezvous in \ASYNCLC in a model with external colors. The algorithm is known to succeed in \ASYNCLC and to fail in \ASYNC \cite{DBLP:conf/opodis/OkumuraWD18}.
\item ``Oku4colsX'': the algorithm achieves rendezvous in \ASYNCLC in a model with external colors. It is quasi-self-stabilizing, meaning in requires the starting colors of the robots to be identical. The algorithm is known to succeed in \ASYNCLC and to fail in \ASYNC \cite{DBLP:conf/opodis/OkumuraWD18}.
\item ``Oku3colsX'': the algorithm achieves rendezvous in \ASYNCLC in a model with external colors. It is a rigid, non-self-stabilizing algorithm. The algorithm is known to succeed in \ASYNCLC and to fail in \ASYNC~\cite{DBLP:conf/opodis/OkumuraWD18}.

\end{enumerate}

\subsection{Verification by Model Checking}

Given a rendezvous algorithm and a verification model, the SPIN model checker essentially verifies that the following liveness property (expressed in LTL) holds in every possible execution:
\begin{quote}
\begin{verbatim}
ltl gathering { <> [] (position == SAME) }
\end{verbatim}
\end{quote}

The formula defines a predicate called \verb|gathering| with the meaning that there is a time after which the position is always \SAME. Concretely, to verify the property, the model checker runs an exhaustive search in the transition graph of configurations such that all initial configurations lead to some cycle such that the predicate \verb|gathering| holds or, in other words, that the variable \verb|position| is equal to \SAME in every configuration of such cycle(s).

\newcommand\OK{\checkmark}
\newcommand\FAIL{-}
\newcommand\slanted[1]{\multicolumn{1}{l}{\rlap{\rotatebox{45}{#1}}}}

%
%

\begin{table}
	\centering
    \small
	\caption{Results of model-checking liveness}
	
	\renewcommand{\arraystretch}{0.9}
	\begin{tabular}{ccccccl}
	&
	&
	& LC-atomic
	& Move-atomic
	&
	\\
	  \CENTRALIZED
	& \FSYNC
	& \SSYNC
	& \ASYNC
	& \ASYNC
	& \ASYNC
	& 
	\\ \hline
\multicolumn{7}{c}{Non-Rigid Self-Stabilizing}
	\\ \hline
	  \FAIL & \FAIL & \FAIL & \FAIL & \FAIL & \FAIL
	& NoMove
	\\
	  \FAIL & \OK & \FAIL & \FAIL & \FAIL & \FAIL
	& ToHalf
	\\
	  \OK & \FAIL & \FAIL & \FAIL & \FAIL & \FAIL
	& ToOther
	\\
	  \OK & \OK & \OK & \OK & \FAIL & \FAIL
	& Vig2Cols
	\\
	  \OK & \OK & \OK & \OK & \OK & \OK
	& Vig3Cols
	\\
	  \OK & \OK & \OK & \OK & \OK & \OK
	& Her2Cols
	\\
	  \OK & \OK & \OK & \FAIL & \FAIL & \FAIL
	& Flo3ColsX
	\\
	  \OK & \OK & \OK & \OK & \FAIL & \FAIL
	& Oku5ColsX \\
	  \OK & \FAIL & \FAIL & \FAIL & \FAIL & \FAIL
	& Oku4ColsX \\
	  \OK & \FAIL & \FAIL & \FAIL & \FAIL & \FAIL
	& Oku3ColsX
	\\ \hline	
\multicolumn{7}{c}{Rigid\footnotemark[3] Quasi-Self-Stabilizing}
	\\ \hline
	  \OK & \OK & \OK & \OK & \FAIL & \FAIL
	& Oku4Cols QSS\footnotemark[3]
	\\ \hline

\multicolumn{7}{c}{Rigid Non-Self-Stabilizing}
	\\ \hline
	  \OK & \OK & \OK & \OK & \FAIL & \FAIL
	& Oku3Cols NSS
	\\ \hline
	\end{tabular}
\end{table}

\addtocounter{footnote}{1}
\footnotetext{This algorithm is supposed to be non-rigid. However, as is shown in the appendix, proving non-rigid, non-self-stabilizing algorithms is a lot trickier and cannot be done by the current version of our model checker.}

Two of those results were actually unexpected: Oku3colsX~\cite{DBLP:conf/opodis/OkumuraWD18} and Oku4colsX~\cite{DBLP:conf/opodis/OkumuraWD18} are not supposed to be self-stabilizing at all, yet are verified to be self-stabilizing under the centralized scheduler by our model checker. However, looking in details at the algorithms, it turns out that the key counter-examples to self-stabilization rely on a simultaneous execution of both robots, which explains the result.

\section{Conclusion and Future Work}

In this paper, we introduced the first model for continuous space rendezvous algorithms that enables mechanical verification. To achieve this, we designed a new abstraction layer for mobile robots and proved that our new model was sufficient for proving properties that are relevant for our purpose. This lead to unexpected but interesting results for general self-stabilizing robot algorithms, such as the non existence of rigid-only self-stabilizing algorithms, and the possible emergence of memory in a system of oblivious robots due to \ASYNC scheduling.
Using SPIN as a basis for our work, we were able to confirm known results for ten different rendezvous algorithms proposed in the literature (performance results are presented in Table~\ref{fig:ModChe}). 
\\
Our framework is, however still incomplete and ongoing. While we describe a way to extend our model to non-rigid non-self-stabilizing algorithms, the actual implementation is not complete yet. 
Also, although the proof arguments sustaining our abstract model are relatively simple, for the sake of having a completely mechanically verified result, we would like to have rewrite them formally using a proof assistant such as CoQ~\cite{PotopSTU19}.

\begin{table}
	\centering
	\caption{Model Checker runtime performance for the \ASYNC scheduler on an Intel i7-8650U running SPIN 6.4.9 on Arch Linux.}
	\label{fig:ModChe}
	\begin{tabular}{lrrrrrr}
	    & \multicolumn{2}{c}{States} & & & Runtime & Memory
	  \\\cline{2-3}
	    & Sorted & Matched & Transitions & Atomic steps &
	       \multicolumn{1}{c}{[ms]} & \multicolumn{1}{c}{[MB]}
	  \\ \hline
	  NoMove
	  &  12,080 &      2,738 &     17,604 &     73,027 &   220 & 145
	  \\
	  ToHalf
	  &   6,141 &        979 &     10,979 &     57,152 &   130 & 135
	  \\
	  ToOther
	  &   4,046 &         57 &      8,367 &     34,372 &   110 & 134
	  \\
	  Vig2Cols
	  & 188,010 &  4,014,448 &  5,452,656 & 33,925,728 &  3,110 & 151
	  \\
	  Vig3Cols
	  & 612,209 & 13,678,976 & 18,419,016 & $114 \times 10^6$ & 11,200 & 190
	  \\
	  Her2Cols
	  & 395,150 &  8,589,648 & 11,652,481 & 72,971,392 &  6,840 & 170
	  \\
	  Flo3ColsX
	  &  13,053 &     48,509 &     80,419 &    440,286 &   210 & 135
	  \\
	  Oku5ColsX
	  & 414,247 &  8,981,645 & 12,126,155 & 73,027,637 &  7,870 & 172
	  \\
	  Oku4ColsX
	  & 307,795 &  6,607,778 &  8,936,310 & 54,718,251 &  5,080 & 162
	  \\
	  Oku3ColsX
	  &  83,072 &  1,714,653 &  2,329,400 & 14,330,584 &  1,380 & 142
	  \\
	  Oku4Cols QSS
	  & 307,793 &  6,607,778 &  8,936,308 & 54,718,251 &  5,110 & 162
	  \\
	  Oku3Cols NSS
	  &  83,070 &  1,714,653 &  2,329,398 & 14,330,584 &  1,380 & 142
	  \\ \hline
	\end{tabular}
\end{table}

\newpage

{\small

\begin{thebibliography}{10}

\bibitem{DBLP:journals/siamcomp/AgmonP06}
Noa Agmon and David Peleg.
\newblock Fault-tolerant gathering algorithms for autonomous mobile robots.
\newblock {\em {SIAM} J. Comput.}, 36(1):56--82, 2006.

\bibitem{DBLP:journals/trob/AndoOSY99}
Hideki Ando, Yoshinobu Oasa, Ichiro Suzuki, and Masafumi Yamashita.
\newblock Distributed memoryless point convergence algorithm for mobile robots
  with limited visibility.
\newblock {\em {IEEE} Trans. Robotics and Automation}, 15(5):818--828, 1999.

\bibitem{AugerBCTU13}
C{\'{e}}dric Auger, Zohir Bouzid, Pierre Courtieu, S{\'{e}}bastien Tixeuil, and
  Xavier Urbain.
\newblock Certified impossibility results for byzantine-tolerant mobile robots.
\newblock In {\em Proc. 15th Intl. Symp. on Stabilization, Safety, and Security
  of Distributed Systems {(SSS)}}, pages 178--190, November 2013.

\bibitem{BalabonskiCPRTU18}
Thibaut Balabonski, Pierre Courtieu, Robin Pelle, Lionel Rieg, S{\'{e}}bastien
  Tixeuil, and Xavier Urbain.
\newblock Brief announcement: Continuous vs. discrete asynchronous moves: {A}
  certified approach for mobile robots.
\newblock In {\em Proc. 20th Intl. Symp. on Stabilization, Safety, and Security
  of Distributed Systems {(SSS)}}, pages 404--408, November 2018.

\bibitem{BalabonskiPRT18}
Thibaut Balabonski, Robin Pelle, Lionel Rieg, and S{\'{e}}bastien Tixeuil.
\newblock A foundational framework for certified impossibility results with
  mobile robots on graphs.
\newblock In {\em Proc. 19th intl. Conf. on Distributed Computing and
  Networking, {(ICDCN)}}, pages 5:1--5:10, January 2018.

\bibitem{berard15infsoc}
B{\'e}atrice B{\'e}rard, Pierre Courtieu, Laure Millet, Maria Potop-Butucaru,
  Lionel Rieg, Nathalie Sznajder, S{\'e}bastien Tixeuil, and Xavier Urbain.
\newblock {[Invited Paper] Formal Methods for Mobile Robots: Current Results
  and Open Problems}.
\newblock {\em {International Journal of Informatics Society}}, 7(3):101--114,
  2015.

\bibitem{BerardLMPTT16}
B{\'{e}}atrice B{\'{e}}rard, Pascal Lafourcade, Laure Millet, Maria
  Potop{-}Butucaru, Yann Thierry{-}Mieg, and S{\'{e}}bastien Tixeuil.
\newblock Formal verification of mobile robot protocols.
\newblock {\em Distributed Computing}, 29(6):459--487, 2016.

\bibitem{BonnetDPPT14}
Fran{\c{c}}ois Bonnet, Xavier D{\'{e}}fago, Franck Petit, Maria
  Potop{-}Butucaru, and S{\'{e}}bastien Tixeuil.
\newblock Discovering and assessing fine-grained metrics in robot networks
  protocols.
\newblock In {\em Proc. 33rd {IEEE} Intl. Symp. on Reliable Distributed Systems
  Workshops, ({SRDS} Workshops)}, pages 50--59, October 2014.

\bibitem{DBLP:conf/icdcs/Bouzid0T13}
Zohir Bouzid, Shantanu Das, and S{\'{e}}bastien Tixeuil.
\newblock Gathering of mobile robots tolerating multiple crash faults.
\newblock In {\em Proc. 33rd {IEEE} Intl. Conf. on Distributed Computing
  Systems {(ICDCS)}}, pages 337--346, July 2013.

\bibitem{DBLP:conf/algosensors/CiceroneSN14}
Serafino Cicerone, Gabriele~Di Stefano, and Alfredo Navarra.
\newblock Minimum-traveled-distance gathering of oblivious robots over given
  meeting points.
\newblock In {\em Proc. 10th Intl. Symp. on Algorithms and Experiments for
  Sensor Systems, Wireless Networks and Distributed Robotics, {(ALGOSENSORS)},
  Revised Selected Papers}, pages 57--72, September 2014.

\bibitem{DBLP:journals/siamcomp/CohenP05}
Reuven Cohen and David Peleg.
\newblock Convergence properties of the gravitational algorithm in asynchronous
  robot systems.
\newblock {\em {SIAM} J. Comput.}, 34(6):1516--1528, 2005.

\bibitem{CourtieuRTU15}
Pierre Courtieu, Lionel Rieg, S{\'{e}}bastien Tixeuil, and Xavier Urbain.
\newblock Impossibility of gathering, a certification.
\newblock {\em Inf. Process. Lett.}, 115(3):447--452, 2015.

\bibitem{CourtieuRTU16}
Pierre Courtieu, Lionel Rieg, S{\'{e}}bastien Tixeuil, and Xavier Urbain.
\newblock Certified universal gathering in $\mathbb{R}^2$ for oblivious mobile
  robots.
\newblock In {\em Proc. 30th Intl. Symp. on Distributed Computing {(DISC)}},
  pages 187--200, September 2016.

\bibitem{0001FPSY16}
Shantanu Das, Paola Flocchini, Giuseppe Prencipe, Nicola Santoro, and Masafumi
  Yamashita.
\newblock Autonomous mobile robots with lights.
\newblock {\em Theor. Comput. Sci.}, 609:171--184, 2016.

\bibitem{DBLP:conf/wdag/DefagoGMP06}
Xavier D{\'{e}}fago, Maria Gradinariu, St{\'{e}}phane Messika, and
  Philippe~Raipin Parv{\'{e}}dy.
\newblock Fault-tolerant and self-stabilizing mobile robots gathering.
\newblock In {\em Proc. 20th Intl. Symp. on Distributed Computing {(DISC)}},
  pages 46--60, September 2006.

\bibitem{DevismesLPRT12}
St{\'{e}}phane Devismes, Anissa Lamani, Franck Petit, Pascal Raymond, and
  S{\'{e}}bastien Tixeuil.
\newblock Optimal grid exploration by asynchronous oblivious robots.
\newblock In {\em Proc. 14th Intl. Symp. on Stabilization, Safety, and Security
  of Distributed Systems {(SSS)}}, pages 64--76, October 2012.

\bibitem{DoanBO16}
Ha~Thi~Thu Doan, Fran{\c{c}}ois Bonnet, and Kazuhiro Ogata.
\newblock Model checking of a mobile robots perpetual exploration algorithm.
\newblock In {\em Proc. 6th Intl. Workshop on Structured Object-Oriented Formal
  Language and Method {(SOFL+MSVL)}, Revised Selected Papers}, pages 201--219,
  November 2016.

\bibitem{Doan0017}
Ha~Thi~Thu Doan, Fran{\c{c}}ois Bonnet, and Kazuhiro Ogata.
\newblock Model checking of robot gathering.
\newblock In {\em Proc. 21st Intl. Conf. on Principles of Distributed Systems,
  {(OPODIS)}}, pages 12:1--12:16, December 2017.

\bibitem{2019Flocchini}
Paola Flocchini, Giuseppe Prencipe, and Nicola~Santoro (Eds.).
\newblock {\em Distributed Computing by Mobile Entities-Current Research in
  Moving and Computing}.
\newblock Lecture Notes in Computer Science, 11340. Springer, 2019.

\bibitem{2012Flocchini}
Paola Flocchini, Giuseppe Prencipe, and Nicola Santoro.
\newblock {\em Distributed Computing by Oblivious Mobile Robots}.
\newblock Synthesis Lectures on Distributed Computing Theory. Morgan {\&}
  Claypool Publishers, 2012.

\bibitem{DBLP:journals/tcs/FlocchiniPSW05}
Paola Flocchini, Giuseppe Prencipe, Nicola Santoro, and Peter Widmayer.
\newblock Gathering of asynchronous robots with limited visibility.
\newblock {\em Theor. Comput. Sci.}, 337(1-3):147--168, 2005.

\bibitem{DBLP:journals/tcs/FlocchiniSVY16}
Paola Flocchini, Nicola Santoro, Giovanni Viglietta, and Masafumi Yamashita.
\newblock Rendezvous with constant memory.
\newblock {\em Theor. Comput. Sci.}, 621:57--72, 2016.

\bibitem{HeribanDT18}
Adam Heriban, Xavier D{\'{e}}fago, and S{\'{e}}bastien Tixeuil.
\newblock Optimally gathering two robots.
\newblock In {\em Proc. 19th Intl. Conf. on Distributed Computing and
  Networking, {ICDCN}}, pages 3:1--3:10, January 2018.

\bibitem{DBLP:journals/siamcomp/IzumiSKIDWY12}
Taisuke Izumi, Samia Souissi, Yoshiaki Katayama, Nobuhiro Inuzuka, Xavier
  D{\'{e}}fago, Koichi Wada, and Masafumi Yamashita.
\newblock The gathering problem for two oblivious robots with unreliable
  compasses.
\newblock {\em {SIAM} J. Comput.}, 41(1):26--46, 2012.

\bibitem{MilletPST14}
Laure Millet, Maria Potop{-}Butucaru, Nathalie Sznajder, and S{\'{e}}bastien
  Tixeuil.
\newblock On the synthesis of mobile robots algorithms: The case of ring
  gathering.
\newblock In {\em Proc. 16th Intl. Symp. on Stabilization, Safety, and Security
  of Distributed Systems {(SSS)}}, pages 237--251, September 2014.

\bibitem{DBLP:conf/opodis/OkumuraWD18}
Takashi Okumura, Koichi Wada, and Xavier D{\'{e}}fago.
\newblock Optimal rendezvous $\mathcal{L}$-algorithms for asynchronous mobile
  robots with external-lights.
\newblock In {\em Proc. 22nd Intl. Conf. on Principles of Distributed Systems
  {(OPODIS)}}, pages 24:1--16, December 2018.

\bibitem{OkumuraWK17}
Takashi Okumura, Koichi Wada, and Yoshiaki Katayama.
\newblock Brief announcement: Optimal asynchronous rendezvous for mobile robots
  with lights.
\newblock In {\em Proc. 19th Intl. Symp. on Stabilization, Safety, and Security
  of Distributed Systems {(SSS)}}, pages 484--488, November 2017.

\bibitem{DBLP:conf/birthday/OkumuraWK18}
Takashi Okumura, Koichi Wada, and Yoshiaki Katayama.
\newblock Rendezvous of asynchronous mobile robots with lights.
\newblock In {\em Adventures Between Lower Bounds and Higher Altitudes - Essays
  Dedicated to Juraj Hromkovi{\v{c}} on the Occasion of His 60th Birthday},
  pages 434--448, 2018.

\bibitem{PotopSTU19}
Maria Potop{-}Butucaru, Nathalie Sznajder, S{\'{e}}bastien Tixeuil, and Xavier
  Urbain.
\newblock Formal methods for mobile robots.
\newblock In Paola Flocchini, Giuseppe Prencipe, and Nicola Santoro, editors,
  {\em Distributed Computing by Mobile Entities, Current Research in Moving and
  Computing.}, volume 11340 of {\em LNCS}, pages 278--313. Springer, 2019.

\bibitem{RubinZMA15}
Sasha Rubin, Florian Zuleger, Aniello Murano, and Benjamin Aminof.
\newblock Verification of asynchronous mobile-robots in partially-known
  environments.
\newblock In {\em Proc. 18th Intl. Conf. on Principles and Practice of
  Multi-Agent Systems {(PRIMA)}}, pages 185--200, October 2015.

\bibitem{SangnierSPT17}
Arnaud Sangnier, Nathalie Sznajder, Maria Potop{-}Butucaru, and S{\'{e}}bastien
  Tixeuil.
\newblock Parameterized verification of algorithms for oblivious robots on a
  ring.
\newblock In Daryl Stewart and Georg Weissenbacher, editors, {\em 2017 Formal
  Methods in Computer Aided Design, {FMCAD} 2017, Vienna, Austria, October 2-6,
  2017}, pages 212--219. {IEEE}, 2017.

\bibitem{DBLP:journals/taas/SouissiDY09}
Samia Souissi, Xavier D{\'{e}}fago, and Masafumi Yamashita.
\newblock Using eventually consistent compasses to gather memory-less mobile
  robots with limited visibility.
\newblock {\em {TAAS}}, 4(1):9:1--9:27, 2009.

\bibitem{SuzukiY99}
Ichiro Suzuki and Masafumi Yamashita.
\newblock Distributed anonymous mobile robots: Formation of geometric patterns.
\newblock {\em {SIAM} J. Comput.}, 28(4):1347--1363, 1999.

\bibitem{Viglietta13}
Giovanni Viglietta.
\newblock Rendezvous of two robots with visible bits.
\newblock In {\em Proc. 9th Intl. Symp. on Algorithms and Experiments for
  Sensor Systems, Wireless Networks and Distributed Robotics, {(ALGOSENSORS)}},
  pages 291--306, September 2013.

\end{thebibliography}

}

\clearpage
\appendix
\section{Appendix}
\subsection{Reminder: Robot System Model}


An execution of robot $r$ is defined as a possibly infinite sequence of activation cycles of $r$.
A global execution is a sequence of events on both robots, such that the event of each robot $r$ is a robot execution of $r$, and the interleaving of events follows the rules of the activation model:
\begin{description}
\item[\CENTRALIZED]
	The activation cycle of a robot is atomic. In other words, a single robot is activated at a time and executes a full activation cycle each time (see Figure~\ref{fig:sched:centralized}).
\item[\FSYNC]
	The activation cycles of both robots are executed simultaneously. Equivalently, the robots always follows the following atomic sequence: each robot executes a \LOOK event and then, in turn, each robot sequentially execute \COMPUTE, \BEGMOVE, and \ENDMOVE. This is depicted in Figure~\ref{fig:sched:FSYNC}.
\item[\SSYNC]
	Activation cycles can be either centralized or combined (as \FSYNC).	
\item[\ASYNC]
	Each event is atomic but there is no atomicity between events (see Figure~\ref{fig:sched:ASYNC}).
\item[\ASYNCLC]
	Same as \ASYNC, but the \LOOK and \COMPUTE events execute atomically.
\item[\ASYNCMOVE]
	Same as \ASYNC, but the \BEGMOVE and \ENDMOVE events execute atomically.
\end{description}

In this paper, we assume that robots are always activable and that the scheduling is fair. Consequently, both robots are activated infinitely many times.

\subsection{Proving lemmas and theorems of Section~\ref{sec:theorems}}

\begin{proof}[Proof: Lemma \ref{Lem:ComToSim}]
Since the set of initial configurations allowed for \emph{simple} self-stabilization is a subset of the one allowed for \emph{complete} self-stabilization, if all \emph{complete} initial configurations lead to a successful execution, then all \emph{simple} starting execution also do.
\end{proof}

\begin{proof}[Proof: Lemma \ref{Lem:SimToComSSYNC}]
Under the \FSYNC, \CENTRALIZED and \SSYNC scheduler, any \emph{complete} initial configuration becomes a \emph{simple} initial configuration after all robots finish their current cycle.
\end{proof}

\begin{proof}[Proof: Theorem \ref{Thm:ObliviousMemory} : Oblivious Robots]
Let us assume two robots $A$ and $B$ running the ToOther algorithm, \emph{i.e.}, the target is always the other robot. Let us also assume an initial configuration where $A$ is in the \WAIT phase and $B$ is in the \COMPUTE phase and has targeted point $P_1$ such that $|AP_1|=|BP_1|=|AB|$. In other words, $ABP_1$ is an equilateral triangle. Note that this is a \emph{complete}, but not \emph{simple} initial configuration.

We first activate $A$, which is now in its \COMPUTE phase and targets the current location of $B$. We then activate $B$ which starts moving towards $P_1$. We then activate $B$ again as it reaches $P_1$ and is now in its \WAIT phase.

This current configuration is identical to the initial configuration. Thus, we have an execution that repeats infinitely often while never being reachable from both robots starting in the \WAIT phase, because starting from the \WAIT phase cannot yield a target outside the $ [AB] $ segment

While this is the simplest example of the behaviour that we could devise, it should be noted that a similar execution could be achieved using a move-to-half algorithm, the only difference being the $ABP_1$ would be shrinking with each activation. Similarly, any initial configuration that included a target outside of the line $(AB)$ could lead to a similar execution, given the right algorithm and ASYNC scheduling. 

Figure \ref{fig:Oblivious} visually shows the execution, with the colored cross showing the target of the robot.

\begin{figure}
	\centering
	\includegraphics[width=\columnwidth]{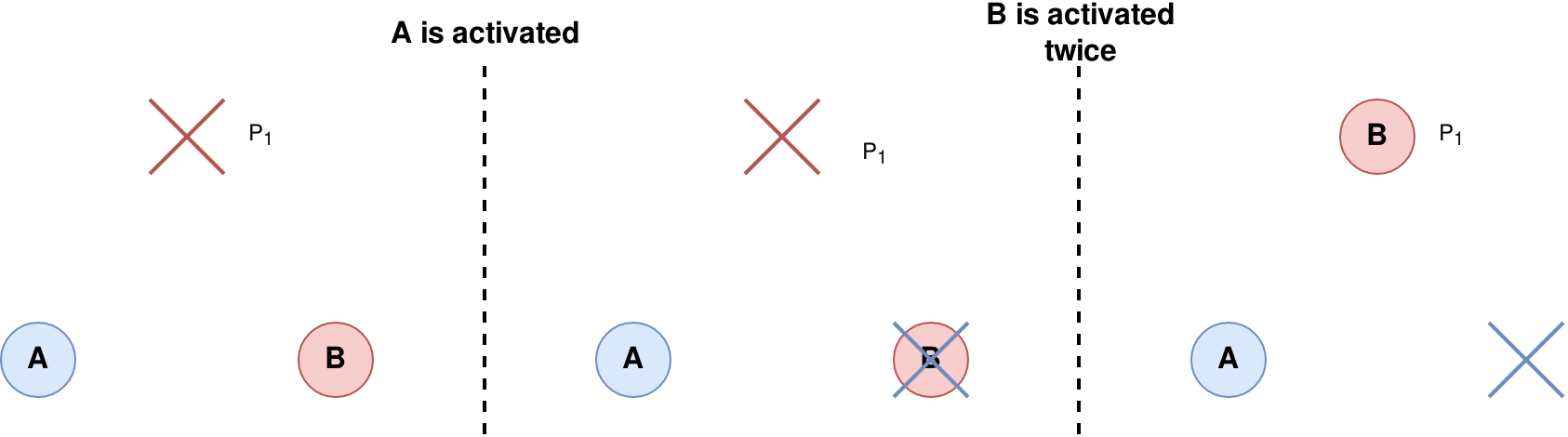}
	\caption{Proof for Oblivious \emph{Complete} Self-Stabilization}
	\label{fig:Oblivious}
\end{figure}

\end{proof}

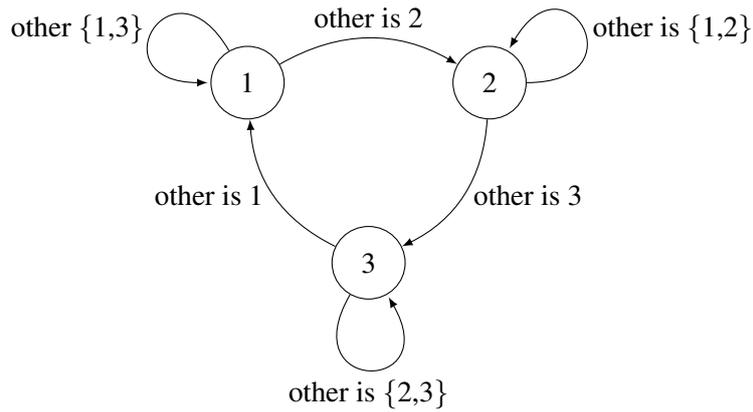
\begin{figure}
	\centering
	\begin{tikzpicture}
		\node[state] (B) {1};
		\node[state] (W) [right of=B] {2};	
		\node[state] (R) at (0.1\columnwidth,-0.15\columnwidth) {3};	
		\path[->] (B) edge[bend left] node[above,align=center]{other is 2} (W);
		\path[->] (W) edge[bend left] node[right,align=center]{other is 3} (R);
		\path[->] (R) edge[bend left] node[left,align=center]{other is 1} (B);

		\path[->] (W) edge[out=0,in=60,loop] node[right,align=center]{other is \{1,2\}} (W);
		\path[->] (B) edge[out=120,in=180,loop] node[left,align=center]{other \{1,3\}} (B);
    	\path[->] (R) edge[out=240,in=300,loop] node[below,align=center]{other is \{2,3\}} (R);
	\end{tikzpicture}
	\caption{Algorithm for the Luminous \emph{Complete} Self-Stabilization}
	\label{fig:colors}
\end{figure}

\begin{proof}[Proof: Theorem \ref{Thm:ObliviousMemory}: Luminous Robots]
Let us assume the three-color algorithm in Figure~\ref{fig:colors}, and an initial configuration of two robots $A$ and $B$, where $A$ starts in color~1, in the \WAIT phase, and $B$ in color~2, in the \COMPUTE phase, with~3 as a pending color.

We then follow the execution described in Figure \ref{fig:ColExec}. We see that the last configuration is identical to the first one with $A$ and $B$ swapped, which means the execution can be executed infinitely often. 

\begin{figure}
	\centering
	\includegraphics[width=.9\columnwidth]{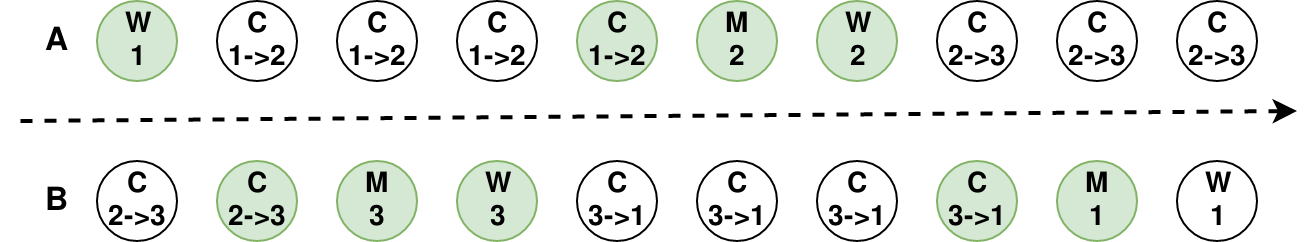}
	\caption{Execution for the Luminous proof of Theorem \ref{Thm:ObliviousMemory}. W indicates the \WAIT phase, C \COMPUTE, and M \MOVE. 1, 2 and 3 indicate the color and, if applicable, the arrow indicates a pending color. The activated robot is highlighted in green.}
	\label{fig:ColExec}
\end{figure}

Let us now prove that this execution cannot be reached from a \WAIT/\WAIT initial configuration. 

A three-color algorithm allows for 6 different color combinations.
\begin{itemize}
    \itemsep0em 
    \item Starting from \{1,2\}, robot $B$ is stuck in color 2 and robot $A$ turns to color 2. So \{1,2\} leads to \{2,2\}
    \item Starting from \{2,3\}, robot $B$ is stuck in color 3 and robot $A$ turns to color 3. So \{2,3\} leads to \{3,3\}
    \item Starting from \{3,1\}, robot $B$ is stuck in color 1 and robot $A$ turns to color 1. So \{3,1\} leads to \{1,1\}
    \item Starting from {1,1}, no robot can change color.
    \item Starting from {2,2}, no robot can change color.
    \item Starting from {3,3}, no robot can change color.
\end{itemize}

We now see that, if starting from \WAIT/\WAIT, no cycle of changing colors can be reached. Therefore, the previously described execution cannot be reached either.
\end{proof}

\begin{Lemma}[Condition of rendezvous in \ASYNC]
\label{Lem:RDVcond}
To achieve rendezvous in \ASYNC, it is necessary that the targets of both robots are eventually always identical.
\end{Lemma}

\begin{proof}[Proof: Lemma \ref{Lem:RDVcond}]
Let us assume that two robots are currently gathered, that is they occupy the same position, and remain there forever.

Let us assume for the purpose of contradiction that at some point during the execution, their targets become different. Then, there exists an execution where their positions become different. Therefore rendezvous is not actually achieved in this execution, a contradiction.
\end{proof}

\begin{Lemma}[Pending Colors]
\label{Lem:PenCol}
    An execution that is not \emph{simple-reachable} cannot contain a configuration in which all robots have identical \emph{pending} and \emph{visible} colors.
\end{Lemma}

\begin{proof}[Proof: Lemma \ref{Lem:PenCol}]
Let us assume a configuration in which all robots have identical pending and visible colors. 

In this configuration, robots can only be in \WAIT, \MOVE or \COMPUTE with a pending color identical to the visible one.
If both robots are in \WAIT, then this can trivially be reached from a \WAIT/\WAIT configuration.
If both robots are in \COMPUTE, then the following activation turns this configuration into a \COMPUTE/\MOVE configuration with no change of color.
If both robots are in \MOVE, then the following activation turns this configuration into a \WAIT/\MOVE configuration.

If one robot is in \WAIT and the other in \COMPUTE, then, after the second robot has reached \WAIT, the configuration can be reached by starting the second robot in its current color and the first robot in its previous color and activating it. 
The same reasoning holds if a robot is in \WAIT and the second in \MOVE or \COMPUTE and \MOVE.
\end{proof}

Based on Lemma~\ref{Lem:RDVcond} and~\ref{Lem:PenCol} above, we can now prove Theorem~\ref{Thm:targetSS}.
\begin{proof}[Proof: Theorem \ref{Thm:targetSS}]
In the case of oblivious robots, we know that rendezvous is not possible in \SSYNC, hence in \ASYNC.
So, proving that an oblivious algorithm is \emph{simply} self-stabilizing, means that it is also \emph{completely} self-stabilizing.

For a luminous algorithm to be \emph{simply} self-stabilizing, the condition must be true for any pair of colors. To obtain an execution similar to the one described in Figure~\ref{fig:Oblivious}, it is necessary for robots to always move to a target that depends on an outdated position of the other robot.
If this behavior is allowed, then gathering is impossible, regardless of the initial target, as robots never have identical targets.
\end{proof}

\begin{proof}[Proof : Theorem \ref{Thm:SoNice}]
We first prove that complete executions that are not simple-reachable must contain pending colors.

We know that any \emph{complete} initial configuration must contain at least one pending color different from the current one to be non simple-reachable. Let us first consider the case where one robot $A$ has different pending and visible colors $A1$ and $A0$, and robot $B$ has identical colors $B0$.

We first show that the first activated robot must be $B$, as we see in Figure~\ref{fig:exec2}.

\begin{figure}[htb]
	\centering
	\begin{tabular}{c|c||c|c}
	    \multicolumn{2}{c}{A} & \multicolumn{2}{c}{B}
	  \\\hline
	  Computed & Visible & Computed & Visible
	  \\\hline
	  $A1$ & $A0$ & $B0$ & $B0$
	  \\
	  $A2(B0)$ & $A1$ & $B0$ & $B0$
	  \\\hline
	\end{tabular}
    \hfill
    \begin{tabular}{c|c||c|c}
	    \multicolumn{2}{c}{A} & \multicolumn{2}{c}{B}
	  \\\hline
	  Computed & Visible & Computed & Visible
	  \\\hline
	  $A1$ & $A1$ & $B0$ & $B0$
	  \\
	  $A2(B0)$ & $A1$ & $B0$ & $B0$
	  \\\hline
	\end{tabular}
	\caption{Activating robot $A$ first (left) leads to a \emph{simple-reachable} configuration (right)}
	\label{fig:exec2}
\end{figure}

We show a possible memory execution in Figure~\ref{fig:exec1}

\begin{figure}[htb]
	\centering
	\begin{tabular}{c|c||c|c}
	    \multicolumn{2}{c}{A} & \multicolumn{2}{c}{B}
	  \\\hline
	  Computed & Visible & Computed & Visible
	  \\\hline
	  $A1$ & $A0$ & $B0$ & $B0$
	  \\
	  $A1$ & $A0$ & $B1(A0)$ & $B0$
	  \\
	  $A2(B0)$ & $A1$ & $B1(A0)$ & $B0$
	  \\
	  $A2(B0)$ & $A1$ & $B2(A1)$ & $B1(A0)$
	  \\
	  $A3(B1(A0))$ & $A2(B0)$ & $B2(A1)$ & $B1(A0)$
	  \\\hline
	\end{tabular}
	\caption{A possible memory execution, note the continuing dependency on colors $A0$, $B0$, and $A1$}
	\label{fig:exec1}
\end{figure}

We now note that, if the algorithm follows the identical color condition, in order for $B1$ to be different from $B0$, we require $A0$ to be identical to $B0$. If we do not force change in color and have $B1 = B0$, the configuration does not change and the first activation of $A$ leads to the counter example shown in Figure~\ref{fig:exec2}.
Because of this, the scheduler now has two choices : Either activate $A$ or $B$. These two executions are shown if Figure~\ref{fig:exec3} and~\ref{fig:exec4}.

\begin{figure}
	\centering
	\hfill
	\begin{minipage}{0.485\linewidth}
	\begin{tabular}{c|c||c|c}
	    \multicolumn{2}{c}{A} & \multicolumn{2}{c}{B}
	  \\\hline
	  Computed & Visible & Computed & Visible
	  \\\hline
	  $A1$ & $A0$ & $B0$ & $B0$
	  \\
	  $A1$ & $A0$ & $B1(A0)$ & $A0$
	  \\
	  $A2(A0)$ & $A1$ & $B1(A0)$ & $A0$
	  \\\hline
	\end{tabular}
	\caption{After activating $B$ then $A$}
	\label{fig:exec3}
	\end{minipage}
    \hfill
	\begin{minipage}{0.485\linewidth}
	\begin{tabular}{c|c||c|c}
	    \multicolumn{2}{c}{A} & \multicolumn{2}{c}{B}
	  \\\hline
	  Computed & Visible & Computed & Visible
	  \\\hline
	  $A1$ & $A0$ & $B0$ & $B0$
	  \\
	  $A1$ & $A0$ & $B1(A0)$ & $A0$
	  \\
	  $A1$ & $A0$ & $B2(A0)$ & $B1(A0)$
	  \\\hline
	\end{tabular}
	\caption{After activating $B$ twice}
	\label{fig:exec4}
	\end{minipage}
	\hfill
\end{figure}

In the case shown in Figure~\ref{fig:exec3}, following the condition leads to $A1$ being identical to $A0$, which is a contraction.
In the case shown in Figure~\ref{fig:exec4}, following the condition leads to $B1$ being identical to $A0$, which we know leads to a \emph{simple-reachable} configuration.

A similar reasoning can be held for the case where both robots have pending colors. 
\end{proof}

\begin{proof}[Proof: Lemma \ref{Lem:Mov}]
First, these three motions are obviously sufficient since they are the only ones used by both Heriban\etal~\cite{HeribanDT18} and Viglietta~\cite{Viglietta13} with two and three colors, respectively.

Next, we prove that it is necessary to use all of these motions to achieve rendezvous. 

Consider the case where both robots are at distinct positions, anonymous and have the same color. As proven by Viglietta~\cite{Viglietta13} in Proposition~4.1: ``We may assume that both robots get isometric snapshots at each cycle, so they both turn the same colors, and compute destination points that are symmetric with respect to their midpoint. If they never compute the midpoint and their execution is rigid and fully synchronous, they never gather.'' Therefore, Move to Half is necessary.

Similarly, consider now a case where snapshots are not isometric because of different colors. Let us assume that their algorithm makes them target any point between them, but not their own positions.

Because of the case where they are both activated at the same time, their targets need to be identical to gather. 
We model this as $\dfrac{D}{x}$ for robot $A$ and $\dfrac{(x-1)D}{x}$ for robot $B$, with $D$ the distance between $A$ and $B$ and $x$ is a real positive number. This is mandatory in the case where $A$ and $B$ are activated at the same time.

However, if $A$ and $B$ are activated sequentially for a full cycle and $x \neq 1$, then $A$ and $B$ have different targets. As long as no motion where $x=1$ exists, no rendezvous can happen if $A$ and $B$ are separated. When $x=1$, $A$ uses Move to Other and B uses Stay and rendezvous can be achieved. Therefore, M2O and Stay are both necessary.
\end{proof}

\begin{Definition}[Static Execution]
We define a static execution as an execution during which the distance between $A$ and $B$ does not decrease.
\end{Definition}

\begin{proof}[Proof: Lemma \ref{Lem:Conv}]
We only consider three types of motion:\nolinebreak[4]
\begin{itemize}\itemsep0em %
\item Stay (\STAY)
\item Move to the midpoint (\HALF)
\item Move to the other robot (\OTHER)
\end{itemize}

First, it is trivial to see that, using these three types of motion, the distance between $A$ and $B$ cannot increase if both robot start in a \WAIT phase.

Next, we look at what happens after each robot completes at least one full cycle. The resulting distance is presenting in Table~\ref{table:distance}.

\begin{table}
    \centering
    \caption{Distance after a full cycle of $A$ and $n$ full cycles of $B$ with an initial distance of $X$}
    \label{table:distance}
    \medskip
    \begin{tabular}{c|c|c|c}
        & $A$ has a pending \STAY & $A$ has a pending \HALF & $A$ has a pending \OTHER
        \\\hline
        $B$ executes $n$ \STAY
        & $X$ & $\left[\dfrac{X}{2} , X - \delta \right]$ & $\left[0 , X - \delta \right]$
        \\ 
        $B$ executes $n$ M2H
        & $\left[ \dfrac{X}{2^n} , X - n \delta \right]$ & $\left[ 0,X - (n+1) \delta \right]$ & $\left[0 , X - \dfrac{X}{2^n}\right]$
        \\ 
        $B$ executes $n$ M2O
        & $\left[ 0 , X - n \delta \right]$ & $\left[ 0, \dfrac{X}{2} \right]$ & $\left[0,X\right]$
        \\ \hline
    \end{tabular}
\end{table}

We note that only three static executions exist: 
\begin{itemize}
    \item No robot moves,
    \item Robot $A$ has a pending \OTHER and robot $B$ reaches robot $A$, either by: 
    \begin{itemize}
        \item executing at least one \OTHER,
        \item executing an infinite number of \HALF,
    \end{itemize}
\end{itemize}

We also can notice that all other cases converge towards a distance of 0. This implies that any algorithm which includes only these three motions and no static execution reduces the distance between the two robots towards zero. 
\end{proof}

\begin{proof}[Proof: Theorem \ref{Thm:rigid}]

For the purpose of contradiction, let us assume the existence of a completely self-stabilizing algorithm that achieves rendezvous with rigid moves, but not with non-rigid moves. This implies that there exists a non-rigid execution that does not lead to rendezvous.

Since the algorithm works in the rigid case, we know that it cannot include any of the static executions we described in Lemma~\ref{Lem:Conv} since they would also fail with rigid motion. 
Since these executions are not part of the algorithm, in the non-rigid case, any execution of the algorithm necessarily decreases the distance towards zero. 
Therefore the distance between two robots $A$ and $B$ eventually becomes less than $\delta$, with $\delta$ being the distance between $A$ and $B$ below which the behavior becomes rigid.

When that is the case, the behavior of the algorithm is strictly the same as the rigid-motion behavior. Since this execution does not achieve rendezvous, this means that there exists a state, which is part of the states of the rigid-motion behavior that does not achieve rendezvous. 
Hence, a contradiction.
\end{proof}

\begin{proof}[Proof: Theorem \ref{Thm:BinSpace}]

We know that a single line is sufficient to model the plane for the rendezvous problem in the general case. 
First, since robots have no common notion of length, the actual distance between the two robots, other than being gathered or not, cannot be used for deciding to move or change their color.
We assume robots only use the three required movements. Because the motion is rigid, robots always reach their destinations. Therefore, if an execution leads to rendezvous, changing the initial distance between the two robots does not change the outcome of the execution.
\end{proof}

\begin{proof}[Proof: Theorem \ref{Thm:sht}]
Consider an algorithm that achieves rendezvous with non-rigid moves when both robots start from color $C_0$, but fails to do so for any other initial color combination.
To achieve rendezvous, the algorithm must work for any initial distance between both robots, including a distance smaller than $\delta$. Therefore, we know that checking the rigid behavior for color $C_0$ is necessary to prove non-rigid behavior.

For instance, Vig2Cols \cite{Viglietta13} achieves rendezvous with rigid moves when starting from color \BLACK. However, we also know this algorithm to fail with non-rigid moves when starting with the same initial colors. Therefore, solely checking the rigid behavior would lead us to incorrectly consider Vig2Cols \cite{Viglietta13} to be a working non-rigid non-self-stabilizing algorithm. Hence, the condition is not sufficient. 
\end{proof}

This happens because, after starting from an arbitrary distance, we cannot make any assumption about which configuration the system is in when the distance between robots becomes smaller than $\delta$. This configuration may indeed depend on the initial distance and the activation schedule.

\begin{proof}[Proof: Theorem \ref{Thm:wada}]
In this particular case, we need to check three things: 
\begin{enumerate}
    \itemsep0em 
    \item That the rigid algorithm achieves rendezvous.
    \item That the farthest non-rigid algorithm leads to the closest.
    \item That the closest non-rigid algorithm leads to the rigid.
\end{enumerate}
The first point is proven in the rigid part of the paper. We now only need to prove the second and third points.

This is easily done by remembering that any algorithm using only the three moves either allows static executions or not. If it does, then the distance is never reduced, and checking at $3\delta$ is equivalent to any other distance. If it does not, then the distance is eventually reduced to zero and hence eventually enters the second behavior.  
Similarly, we check that the second behavior leads to rigid by reducing distance.
We only need to check that this decrease happens once for every initial configuration to ensure that distance is reduced towards rigid behavior.
\end{proof}

\begin{proof}[Proof: Theorem \ref{Thm:Scheduling}]
In the \FSYNC model, the \LOOK, \COMPUTE, \BEGMOVE, and \ENDMOVE phases of all robots are executed simultaneously. However, since \LOOK is a read-only operation, and \COMPUTE a write-only operation, with regards to color, these operations can be executed sequentially as long as no \LOOK happens after a \COMPUTE (\emph{i.e.}, all $read$ operations happen before the first $write$ operation).
A similar reasoning holds for the \BEGMOVE to \ENDMOVE being the beginning and end of a continuous write operation.
\end{proof}

\begin{proof}[Proof: Theorem \ref{Thm:Events}]
We consider two cases: 
\begin{enumerate}
\item In the first case, one event ($E1$) is a $read$ operation (\LOOK), and one event ($E2$) is a $write$ operation (\COMPUTE).
Since the same piece of information is being read and written at the same time, the result of the read operation cannot be determined. 
In the case of mobile robots, where the $write$ operation is a color transition from color $c_1$ to color $c_2$, we assume that the only colors that can be seen by the \LOOK are $c_1$ or $c_2$\footnote{This assumption is analogous to assuming regularity with registers.}. We then need to consider the case where the LOOK saw a $c_1$ (E1 then E2) and the case where it saw a $c_2$ (E2 then E1).

\item In any other case, since no read and write operation are happening at the same time, activating E1 and E2 simultaneously is identical to activating either E1 then E2, or E2 then E1.
\end{enumerate}
\vspace{-2\baselineskip}
\end{proof}

\begin{proof}[Proof: Theorem \ref{Thm:Fair}]
To properly limit the number of activations of the scheduler, we need to ensure that any change in the configuration made by $A$ that may impact  the snapshot of $B$ has been explored.
When performing its cycle, there are two such elements: $A$'s color, and the distance between $A$ and $B$.

Since there are only two distance states \{\SAME, \NEAR\}, a fixed number $N$ of colors, and it takes 4 activations to perform a cycle, limiting the fair scheduler to an 8$N$-fair scheduler, that is, a fair scheduler that can perform at most 8$N$ activations of robot $A$ between two activations of robot $B$ still ensures that every possible snapshot has been explored.
\end{proof}

\section{The case of non-rigid, non-self-stabilizing Algorithms}

Because of Theorem \ref{Thm:sht}, we cannot directly use our model checker in its current form to check this type of algorithms.

This is because a non-rigid algorithm can reach its rigid behavior in an unpredictable configuration. Our method for solving this issue is to notice that the number of those configurations is finite. More precisely, parameters include: 
\begin{itemize}
    \itemsep0em 
\item The current color of each robot
\item The pending color of each robot
\item The phase of each robot
\item The pending move of each robot
\end{itemize}

This means we have a number $N_{conf}$ of possible configurations. 

We use this fact by creating a counter, which starts at 0 and increases up to $N_{conf}$.

We start the validation process in the \FAR state. When executing the first movement resolution that leads to \NEAR, we create two branches: 

\begin{enumerate}
\item In the first branch, the robot reaches \NEAR and we continue the process
\item In the second branch, the robot actually did not reach \NEAR and is kept at \FAR{}. We increment the counter by 1.
\end{enumerate}

We repeat this process until the counter reaches $N_{conf}  - 1$. We have then created $2^{N_{conf}} - 1$ branches to verify. However, we have ensured that any possible rigid motion behavior configurations has been checked. Repeating this process for every possible non-rigid motion initial configuration is enough to ensure that any possibly failing execution would have been detected.

\subsection{Movement Resolution}
The movement resolution rules described in Figure~\ref{tab:movement} (p.~\pageref{tab:movement}) are implemented by the Promela code described in Listing~\ref{lst:movement} below, during the \ENDMOVE phase of the cycle.
\begin{lstlisting}[caption={Movement resolution\label{lst:movement}}]
if
:: (robot[me].is_moving) ->
    local position_t new_position = position;
    assert( robot[me].pending != STAY );
    if
    :: (position == NEAR || position == SAME) ->
        if
        :: (robot[me].pending == MISS) -> 
            { robot[other].pending = MISS } unless (robot[other].pending == STAY);
            new_position = NEAR;
        :: (robot[me].pending == TO_OTHER) ->
            { robot[other].pending = MISS } unless (robot[other].pending == STAY
                                                    || position == SAME);
            new_position = SAME;
        :: (robot[me].pending == TO_HALF) ->
            if 
            :: (robot[other].pending == TO_HALF) -> robot[other].pending = TO_OTHER
            :: (robot[other].pending == STAY)    -> skip /* do nothing */
            :: else                              -> robot[other].pending = MISS
            fi
        :: else -> assert( false )
        fi;
    fi;
    if
    :: (position != new_position) ->
        eventPositionChange:
            position = new_position
    :: else -> skip /* do nothing */
    fi
:: else -> skip
fi;
robot[me].is_moving = false;
robot[me].pending = STAY;
\end{lstlisting}

In Promela, the meaning of \texttt{if} and the guarded actions \texttt{guard -> action} that follow is different from other languages in the sense that, when several guards are enabled, the execution faces a non-deterministic choice and the exploration of the model checker branches into several executions to explore all enabled guards. The guard \texttt{else} is exclusive in that it is enabled only when no other guards are enabled.

\subsection{Verified Algorithms}


\begin{lstlisting}[caption={No Move Algorithm}]
inline Alg_NoMove(obs, command)
{
    command.move		= STAY;
    command.new_color	= BLACK
}
\end{lstlisting}

\begin{lstlisting}[caption={Move to Half Algorithm}]
inline Alg_ToHalf(obs, command)
{
    command.move		= TO_HALF;
    command.new_color	= BLACK
}
\end{lstlisting}

\begin{lstlisting}[caption={Move to Other Algorithm}]
inline Alg_ToOther(obs, command)
{
    command.move		= TO_OTHER;
    command.new_color	= BLACK
}
\end{lstlisting}

\begin{lstlisting}[caption={Viglietta's 2 colors algorithm \cite{Viglietta13} for \ASYNCLC}]
inline Alg_Vig2Cols(obs, command)
{
    command.move      = STAY;
    command.new_color = obs.color.me;
    if
    :: (obs.color.me == BLACK) ->
        if
        :: (obs.color.other == BLACK) -> command.new_color = WHITE
        :: (obs.color.other == WHITE) -> skip
        fi
    :: (obs.color.me == WHITE) ->
        if
        :: (obs.color.other == BLACK) -> command.move = TO_OTHER
        :: (obs.color.other == WHITE) -> command.move = TO_HALF; command.new_color = BLACK
        fi	
    :: else -> command.new_color = BLACK
    fi
}
\end{lstlisting}

\begin{lstlisting}[caption={Viglietta's 3 colors algorithm \cite{Viglietta13} for \ASYNC}]
inline Alg_Vig3Cols(obs, command)
{
    command.move      = STAY;
    command.new_color = obs.color.me;
    if
    :: (obs.color.me == BLACK) ->
        if
        :: (obs.color.other == BLACK)	-> command.move = TO_HALF; command.new_color = WHITE
        :: (obs.color.other == WHITE)	-> command.move = TO_OTHER
        :: (obs.color.other == RED)		-> skip
        fi
    :: (obs.color.me == WHITE) ->
        if
        :: (obs.color.other == BLACK)	-> skip
        :: (obs.color.other == WHITE)	-> command.new_color = RED
        :: (obs.color.other == RED)		-> command.move = TO_OTHER
        fi
    :: (obs.color.me == RED) -> 
        if
        :: (obs.color.other == BLACK)	-> command.move = TO_OTHER
        :: (obs.color.other == WHITE)	-> skip
        :: (obs.color.other == RED)		-> command.new_color = BLACK
        fi
    :: else -> command.new_color = BLACK
    fi
}
\end{lstlisting}

\clearpage
\begin{lstlisting}[caption={Heriban's 2 colors algorithm \cite{HeribanDT18} for \ASYNC}]
inline Alg_Optimal(obs, command)
{
    command.move      = STAY;
    command.new_color = obs.color.me;
    if
    :: (obs.color.me == BLACK) ->
        if
        :: (obs.color.other == BLACK) -> command.new_color = WHITE
        :: (obs.color.other == WHITE) -> skip
        fi
    :: (obs.color.me == WHITE) ->
        if
        :: obs.same_position -> skip
        :: else ->
            if
            :: (obs.color.other == BLACK) -> command.move = TO_OTHER
            :: (obs.color.other == WHITE) -> command.move = TO_HALF; command.new_color = BLACK
            fi
        fi
    :: else -> command.new_color = BLACK
    fi
}
\end{lstlisting}

\begin{lstlisting}[caption={Flocchini's external lights 3 colors algorithm \cite{DBLP:journals/tcs/FlocchiniSVY16} for \SSYNC}]
inline Alg_FloAlgo3Ext(obs, command)
{
    command.move      = STAY;
    command.new_color = obs.color.me;
    if
    :: (obs.color.other == BLACK)	-> command.move = TO_HALF; command.new_color = WHITE
    :: (obs.color.other == WHITE)	-> command.new_color = RED
    :: (obs.color.other == RED)		-> command.move = TO_OTHER; command.new_color = BLACK
    :: else -> command.new_color = BLACK
    fi
}
\end{lstlisting}

\begin{lstlisting}[caption={Okumura's external lights 5 colors algorithm \cite{DBLP:conf/opodis/OkumuraWD18} for \ASYNCLC}]
inline Alg_Wada5Ext(obs, command)
{
    command.move      = STAY;
    command.new_color = obs.color.me;
    if
    :: (obs.color.other == BLACK)	-> command.move = TO_HALF;    command.new_color = WHITE
    :: (obs.color.other == WHITE)	-> command.new_color = RED
    :: (obs.color.other == RED)		-> command.move = TO_OTHER;   command.new_color = YELLOW
    :: (obs.color.other == YELLOW)	-> command.new_color = GREEN
    :: (obs.color.other == GREEN)	-> command.new_color = BLACK
    :: else -> command.new_color = BLACK
    fi
}
\end{lstlisting}

\clearpage
\begin{lstlisting}[caption={Okumura's external lights 4 colors algorithm \cite{DBLP:conf/opodis/OkumuraWD18} for quasi-self-stabilizing \ASYNCLC}]
inline Alg_Oku4ColsX(obs, command)
{
    command.move      = STAY;
    command.new_color = obs.color.me;
    if
    :: (obs.color.other == BLACK)	-> command.move = TO_HALF;    command.new_color = WHITE
    :: (obs.color.other == WHITE)	-> command.new_color = RED
    :: (obs.color.other == RED)		-> command.move = TO_OTHER;   command.new_color = YELLOW
    :: (obs.color.other == YELLOW)	-> command.new_color = BLACK
    :: else -> command.new_color = BLACK
    fi
}
\end{lstlisting}

\begin{lstlisting}[caption={Okumura's external lights 3 colors algorithm \cite{DBLP:conf/opodis/OkumuraWD18} for non-self-stabilizing rigid \ASYNCLC}]
inline Alg_Oku3ColsX(obs, command)
{
    command.move      = STAY;
    command.new_color = obs.color.me;
    if
    :: (obs.color.other == BLACK)	-> command.move = TO_HALF; command.new_color = WHITE
    :: (obs.color.other == WHITE)	-> command.new_color = RED
    :: (obs.color.other == RED)		-> command.move = TO_OTHER; command.new_color = WHITE
    :: else -> command.new_color = BLACK
    fi
}
\end{lstlisting}

\subsection{Verification Model}

\subsubsection{State variables}

The state of the system is represented by the following \emph{explicit} variables:
\begin{itemize}
\item \verb|distance| $\in \{\NEAR, \SAME\}$\\
	The distance between the two robots. \NEAR means that it is equal or smaller and they have distinct positions, and \SAME that they share the same location. When checking non-rigid non-self-stabilizing algorithm, the distance \FAR can also be used.
	
\item \verb|robot[_].color| $\in \{\BLACK, \WHITE, \RED, \YELLOW, \GREEN\}$\\
	The observable color of the robot. The color \RED is used only in 3 colors, 4 colors, and 5 colors algorithms. The color \YELLOW is used only in 4 colors and 5 colors algorithms. The color \GREEN is used only in 5 colors algorithms.

\end{itemize}

and the following \emph{implicit} variables:
\begin{itemize}

\item \verb|robot[_].phase| $\in \{\LOOK, \COMPUTE, \BEGMOVE, \ENDMOVE\}$\\
	Keeps track of the next event to execute in the activation cycle of the robot. This variable is managed by the scheduler and is particularly important for the \ASYNC scheduler. 

\item \verb|robot[_].pending_move| $\in \{\STAY, \OTHER, \HALF, \MISS, \BOT\}$\\
	Holds the movement computed by the robot. This is not observable but used to resolve movements during the \ENDMOVE phases. The variable is updated during the \LOOK phase, based on the movement computed by the algorithm.

\item \verb|robot[_].pending_color| $\in \{\BLACK, \WHITE, \RED, \YELLOW, \GREEN, \BOT\}$\\
	New color for the robot, as computed by the algorithm. The color is computed during the \LOOK phase of the robot and used during the \COMPUTE phase to update the visible color of the robot.

\end{itemize}

Depending on the model and algorithms, this can lead to at most $2 * (5 * 4 * 5 * 6)^2 = 720'000$ different states.


\subsubsection{Activation Phases}

\begin{description}
\item[\LOOK]
	Reads the current state of the environment and saves it as an observation.
	Applies the algorithm to compute the pending move and the new color.
\item[\COMPUTE] 
	Updates the color of the robot in the environment.
\item[\BEGMOVE] 
	Begins moving. Unless the pending move is \STAY, the robot's position is undetermined until \ENDMOVE and any that occurs in the interval causes an automatic \MISS move for the other robot.
\item[\ENDMOVE]
	The pending move is resolved into a new position for the robot and the environment is updated accordingly. 
\end{description}

\subsection{Verification and compile options}

\begin{verbatim}
spin -a -DALGO=ALGORITHM -DSCHEDULER=SCHEDULER MainGathering.pml
clang -DMEMLIM=1024 -DXUSAFE -DNOREDUCE -O2 -w -o pan pan.c
./pan -m100000 -a -f -E -n gathering 
\end{verbatim}

\subsubsection{Verification output}

\paragraph{Vig2Cols in \ASYNC (fails)}
{\small
\begin{alltt}
Depth=   19582 States=    1e+06 Transitions= 3.86e+06 
Memory=   146.311 t=     2.25 R=   4e+05
pan:1: acceptance cycle (at depth 2723)
pan: wrote MainGathering.pml.trail

(Spin Version 6.4.9 -- 17 December 2018)
Warning: Search not completed

Full statespace search for:
        never claim             + (gathering)
        assertion violations    + (if within scope of claim)
        acceptance   cycles     + (fairness enabled)
        invalid end states      - (disabled by -E flag)

State-vector 107 byte, depth reached 25730, \textbf{errors: 1}
   189620 states, stored (1.45028e+06 visited)
  4046596 states, matched
  5496875 transitions (= visited+matched)
 34174756 atomic steps
hash conflicts:     16160 (resolved)

Stats on memory usage (in Megabytes):
   24.413       equivalent memory usage for states (stored*(State-vector 
                + overhead))
   17.660       actual memory usage for states (compression: 72.34%)
                state-vector as stored = 70 byte + 28 byte overhead
  128.000       memory used for hash table (-w24)
    6.104       memory used for DFS stack (-m100000)
  151.682       total actual memory usage



pan: elapsed time 3.25 seconds
pan: rate 446239.69 states/second
\end{alltt}
}

After a failed execution, SPIN provides a (very verbose) counter-example that can be investigated with the following command.
\begin{verbatim}
spin -t MainGathering.pml
\end{verbatim}

\paragraph{Her2Cols in \ASYNC (Succeeds)}
{\small
\begin{alltt}
Depth=   16958 States=    1e+06 Transitions= 3.84e+06 
Memory=   146.115 t=     2.26 R=   4e+05

(Spin Version 6.4.9 -- 17 December 2018)

Full statespace search for:
        never claim             + (gathering)
        assertion violations    + (if within scope of claim)
        acceptance   cycles     + (fairness enabled)
        invalid end states      - (disabled by -E flag)

State-vector 107 byte, depth reached 17016, errors: 0
   232931 states, stored (1.80493e+06 visited)
  5061150 states, matched
  6866078 transitions (= visited+matched)
 42744752 atomic steps
hash conflicts:     30286 (resolved)

Stats on memory usage (in Megabytes):
   29.989       equivalent memory usage for states (stored*(State-vector 
                + overhead))
   21.666       actual memory usage for states (compression: 72.25%)
                state-vector as stored = 70 byte + 28 byte overhead
  128.000       memory used for hash table (-w24)
    6.104       memory used for DFS stack (-m100000)
  155.686       total actual memory usage



pan: elapsed time 4.07 seconds
pan: rate 443471.25 states/second
\end{alltt}
}

\end{document}